\documentclass[12pt]{article}
\usepackage{amsthm,amsfonts,amsmath,amssymb,mathrsfs,url,graphicx}
\usepackage[usenames,dvipsnames]{color}

\title{Multi-Time Schr\"odinger Equations Cannot Contain Interaction Potentials}

\author{
S\"oren Petrat\footnote{Mathematisches Institut,
	Ludwig-Maximilians-Universit\"at, Theresienstr. 39, 80333 M\"unchen, Germany.
	E-mail: petrat@math.lmu.de}\ \ and
Roderich Tumulka\footnote{Department of Mathematics,
     Rutgers University,
     110 Frelinghuysen Road, Piscataway, NJ 08854-8019, USA.
     E-mail: tumulka@math.rutgers.edu}
}
\date{January 24, 2014}

\theoremstyle{plain}\newtheorem{thm}{Theorem}
\theoremstyle{plain}\newtheorem{lem}[thm]{Lemma}

\newcommand{\be}{\begin{equation}}
\newcommand{\ee}{\end{equation}}
\newcommand{\Hilbert}{\mathscr{H}}

\newcommand{\RRR}{\mathbb{R}}
\newcommand{\CCC}{\mathbb{C}}
\newcommand{\scp}[2]{\langle #1|#2 \rangle}

\newcommand{\Laplace}{\Delta}
\newcommand{\art}{\vec{t}}
\newcommand{\vect}[1]{\boldsymbol{#1}}  

\newcommand{\va}{\boldsymbol{a}}
\newcommand{\vb}{\boldsymbol{b}}
\newcommand{\vy}{\boldsymbol{y}}
\newcommand{\valpha}{\boldsymbol{\alpha}}
\newcommand{\vx}{\boldsymbol{x}}
\newcommand{\sS}{\mathscr{S}}
\newcommand{\free}{\mathrm{free}}
\DeclareMathOperator{\divergence}{div}

\begin{document}
\maketitle

\begin{abstract}
Multi-time wave functions are wave functions that have a time variable for every particle, such as $\phi(t_1,\vx_1,\ldots,t_N,\vx_N)$. They arise as a relativistic analog of the wave functions of quantum mechanics but can be applied also in quantum field theory. The evolution of a wave function with $N$ time variables is governed by $N$ Schr\"o\-din\-ger equations, one for each time variable. These Schr\"o\-din\-ger equations can be inconsistent with each other, i.e., they can fail to possess a joint solution for every initial condition; in fact, the $N$ Hamiltonians need to satisfy a certain commutator condition in order to be consistent. While this condition is automatically satisfied for non-interacting particles, it is a challenge to set up consistent multi-time equations with interaction. We prove for a wide class of multi-time Schr\"o\-din\-ger equations that the presence of interaction potentials (given by multiplication operators) leads to inconsistency. We conclude that interaction has to be implemented instead by creation and annihilation of particles, which, in fact, can be done consistently, as we show elsewhere \cite{pt:2013c}. We also prove the following result: When a cut-off length $\delta>0$ is introduced (in the sense that the multi-time wave function is defined only on a certain set of spacelike configurations, thereby breaking Lorentz invariance), then the multi-time Schr\"o\-din\-ger equations with interaction potentials of range $\delta$ are consistent; however, in the desired limit $\delta\to 0$ of removing the cut-off, the resulting multi-time equations are interaction-free, which supports the conclusion expressed in the title.

\medskip

\end{abstract}

\newpage\tableofcontents

\section{Introduction}\label{sec:intro}

This paper belongs to a series of papers \cite{pt:2013c,pt:2013e,pt:2014a,pt:2014b} exploring multi-time wave functions, i.e., wave functions of the form
\be
\phi(t_1,\vx_1,\ldots,t_N,\vx_N)
\ee
that have a separate time variable for each particle and thus $4N$ variables in total ($\vx_j\in\RRR^3$). In other words, these are wave functions on configurations of $N$ points in space-time. Such wave functions arise naturally as a relativistic generalization of the usual (single-time) wave function of non-relativistic quantum mechanics,
\be
\psi(t,\vx_1,\ldots,\vx_N)
\ee
with $3N+1$ variables. Multi-time wave functions can also be defined in the context of quantum field theories (QFTs) that permit a particle--position representation of a quantum state (in Fock space), when we regard also $N$ as variable \cite{pt:2013c,pt:2014a,pt:2014b}.

This paper is mainly about the \emph{consistency condition} for multi-time Schr\"odinger equations, formulated as \eqref{consistency} in Section~\ref{sec:consistency1} below. Apart from a careful derivation of this condition, our main result is that it excludes interaction by means of a potential (i.e., multiplication operator); we conclude that interaction must be formulated, in the multi-time approach, by means of particle creation and annihilation, as we do in \cite{pt:2013c,pt:2014a}. We also describe here a result about the consistency of potentials of limited range $\delta>0$ on multi-time wave functions with a certain kind of length cut-off at the length $\delta$.

Multi-time wave functions deserve study for several reasons. First, while it was unclear until recently whether any consistent set of multi-time Schr\"odinger equations could involve interaction (and thus, whether multi-time wave functions could be of any physical relevance), we now know \cite{pt:2013c,pt:2014a} that relevant interacting QFTs can be consistently reformulated in terms of multi-time wave functions. Thus, multi-time wave functions are not a mere mathematical speculation but rather a new representation of familiar QFTs. This representation, actually a very natural one, can be regarded as the covariant Schr\"odinger-picture particle-position representation; it is related to the Tomonaga--Schwinger representation but conceptually simpler because the latter is defined on the \emph{infinite}-dimensional space of all spacelike hypersurfaces, whereas the domain of the relevant multi-time wave functions has locally \emph{finite} dimension. The multi-time approach emphasizes the similarity of QFT with quantum mechanics, particularly so by expressing the quantum state in terms of a wave function. Another reason why multi-time wave functions deserve study is that they are such a natural concept, the immediate analog of the wave function of quantum mechanics in a relativistic setting, and manifestly covariant objects. Since the multi-time approach is unfamiliar, we take the time here to develop the theory of the consistency condition in some detail. It may seem that the consistency condition is a novel obstacle, even a drawback, that the multi-time approach brings with it. However, we think of the consistency condition as providing us guidance about how the equations of a relativistic QFTs should be set up. For example, the main result of the present paper, that interaction potentials conflict with the consistency condition, tells us that interaction should be incorporated by means of particle creation and annihilation. Likewise, the consistency condition for the multi-time equations considered in \cite{pt:2014a} tells us that a fermion cannot decay into two fermions.

Multi-time wave functions were considered early on in the history of quantum theory (particularly by Dirac \cite{dirac:1932}, Dirac, Fock, and Podolsky \cite{dfp:1932}, and Bloch \cite{bloch:1934}), but have, as far as we know, never been studied comprehensively. We discuss their application to QFT in \cite{pt:2013c,pt:2014a}. Connections with QFT were made also by G\"unther \cite{Gue52} and Schweber \cite[p.~171]{schweber:1961}. Horwitz and Rohrlich \cite{horwitz:1981} suggested considering a wave function of $5N$ variables, which does not seem to yield a viable reformulation of quantum physics. While we consider one Schr\"o\-din\-ger equation for each time variable, Salpeter and Bethe \cite{SB51} (as well as Marx \cite{Marx:1974}) considered a single higher-order equation for a two-time wave function. Several authors \cite{FL70,DV71,DV81,VAC97,Hah09} have proposed Lorentz-invariant sets of equations for two-time wave functions containing interaction terms that neither are potentials (i.e., multiplication operators) nor involve particle creation, but instead are nonlocal in time (see also Section~\ref{sec:perspective}). We compare the status and significance of multi-time formulations in classical and quantum physics in \cite{pt:2013e}.

\bigskip

The remainder of this paper is organized as follows. In the remainder of Section~\ref{sec:intro}, we introduce multi-time Schr\"odinger equations for the time evolution of multi-time wave functions, formulate the consistency condition \eqref{consistency}, and outline our results. In Section~\ref{sec:consistency} we provide precise formulations of the statement that condition \eqref{consistency} is necessary and sufficient for the consistency of the multi-time equations. In Section~\ref{sec:potentials} we provide precise formulations of our results about the inconsistency of interaction potentials. In Section~\ref{sec:deltarange}, we provide a precise formulation of our result about the consistency of potentials with range $\delta$. In Sections~\ref{sec:proofs_consistency}, \ref{sec:proofs_inconsistentV}, and \ref{sec:proof_delta}, we provide the proofs of the statements made in Sections~\ref{sec:consistency}, \ref{sec:potentials}, and \ref{sec:deltarange}, respectively.

\subsection{Multi-Time Evolution}
For the multi-time wave function $\phi$ to be determined by initial data
\be
\phi(0,\vx_1,\ldots,0,\vx_N)
\ee
we need $N$ Schr\"o\-din\-ger equations, one for each time variable (we set $\hbar=1$):
\be\label{phiHj}
i\frac{\partial \phi}{\partial t_j} = H_j \phi
\ee
for $j=1,\ldots,N$.\footnote{Since we have written the equations \eqref{phiHj} in a Hamiltonian form, we have used a particular Lorentz frame in order to refer to a time variable $t_j$. This is convenient for us at this point but not necessary; $H_j$ will contain derivatives with respect to the spacelike components $\vx_j$, and those could be moved to the left-hand side to write \eqref{phiHj} in a manifestly covariant form, just like the one-particle Dirac equation can be written either in the Hamiltonian form $i\partial \psi/\partial t = H\psi$ with (setting $c=1$) $H=-i\valpha\cdot\nabla+\beta m$ or in the manifestly covariant form $i\gamma^\mu\partial_\mu \psi = m\psi$.} We call equations of the form \eqref{phiHj} \emph{multi-time Schr\"o\-din\-ger equations}\footnote{The expression \emph{Schr\"o\-din\-ger equation} is not meant to imply that $H_j$ involves the Laplace operator, but is understood as including, e.g., the Dirac equation.} or simply \emph{multi-time equations}, and the operators $H_j$ \emph{partial Hamiltonians}.
The connection between the multi-time wave function $\phi$ and the single-time wave function $\psi$ is that on configurations of $N$ space-time points that are simultaneous with respect to the Lorentz frame $L$ to which $\psi$ refers, $\phi$ coincides with $\psi$; i.e.,
\be\label{phipsi}
\phi(t,\vx_1,\ldots,t,\vx_N)=\psi(t,\vx_1,\ldots,\vx_N)\,.
\ee
It follows from \eqref{phipsi} that, at every configuration that is simultaneous with respect to the Lorentz frame $L$,
\be\label{HjH}
\sum_{j=1}^N H_j = H\,,
\ee
where $H$ is the Hamiltonian governing $\psi$,
\be\label{psiH}
i\frac{\partial \psi}{\partial t} = H\psi\,.
\ee

\subsection{Consistency Condition}\label{sec:consistency1}

A novel feature of multi-time equations such as \eqref{phiHj}, absent from the single-time Schr\"o\-din\-ger equation \eqref{psiH}, is that the multi-time equations can be \emph{inconsistent}; being inconsistent means that they possess no non-zero joint solutions $\phi$, or possess non-zero joint solutions only for special initial conditions. The condition for consistency (or integrability), which Bloch \cite{bloch:1934} was already aware of, reads
\be\label{consistency}
\biggl[i\frac{\partial}{\partial t_j} - H_j, i\frac{\partial}{\partial t_k} -H_k\biggr] = 0
\quad \forall j\neq k\,.
\ee

To begin to understand where \eqref{consistency} comes from, consider the multi-time equations \eqref{phiHj} for $N=2$ particles and the simple case where $H_1$ and $H_2$ are time-independent and bounded operators on a Hilbert space $\Hilbert$, e.g., $\Hilbert=L^2(\RRR^3\times\RRR^3)$. We can regard $\phi$ as an $\Hilbert$-valued function of $t_1$ and $t_2$. Then, for arbitrary initial conditions $\phi(0,0)$ we can obtain a solution of the multi-time equations \eqref{phiHj} in two different ways,
\be
\phi(t_1,t_2) = e^{-iH_2t_2}\phi(t_1,0) = e^{-iH_2t_2} e^{-iH_1t_1} \phi(0,0)
\ee
and
\be
\phi(t_1,t_2) = e^{-iH_1t_1}\phi(0,t_2) = e^{-iH_1t_1} e^{-iH_2t_2} \phi(0,0).
\ee
Both expressions agree (and thus yield a joint solution $\phi$) for all initial $\phi(0,0)\in\Hilbert$ if and only if 
\be\label{H1H2=0}
[H_1,H_2]=0\,.
\ee
Since the Hamiltonians do not depend on $t_1$ or $t_2$, they commute with $\partial/\partial t_j$, and the consistency condition \eqref{consistency} amounts to \eqref{H1H2=0}. A more general and detailed derivation and discussion of condition \eqref{consistency} is given in Section~\ref{sec:consistency}.

\subsection{Perspective}
\label{sec:perspective}

The time evolution of non-interacting particles trivially satisfies \eqref{consistency}. The central claim of this paper is that every interaction potential violates the consistency condition \eqref{consistency}.

Before explaining the details of the claim, let us put it into perspective. The claim might be surprising from the point of view of non-relativistic quantum mechanics because there, potentials are the only method of implementing interaction between the particles. A different method, however, is available in quantum field theory, where the particle number is not fixed: there, particles can interact by emitting and absorbing other particles. And indeed, as we show in \cite{pt:2013c,pt:2014a,pt:2014b}, this kind of interaction can be implemented with consistent multi-time equations; the analysis of the consistency of these equations is more delicate because, as the number of particles is not fixed, also the number of time variables is not fixed; yet, the analysis can be done and confirms the consistency of natural choices of multi-time equations with particle creation and annihilation. 

From a different perspective, our claim, that interaction potentials make multi-time equations inconsistent, may be \emph{un}surprising: Multi-time wave functions were introduced for the purpose of a \emph{covariant} description of the quantum state, and interaction potentials may seem incompatible with relativity. After all, interaction potentials involve a function $V(x,y)$ with $x,y$ different space-time points, and this suggests that interaction potentials represent a direct (and possibly faster-than-light) action-at-a-distance, in conflict with principles of relativity and presumably implying, in particular, the possibility of superluminal signaling.

On the other hand, it is not uncommon to use, in (relativistic) quantum electrodynamics, the Coulomb gauge of the quantized fields, which leads to an explicit Coulomb potential term in the Hamiltonian besides the quantized degrees of freedom of the electromagnetic field; see, e.g., \cite{CDG:1989}. So it is perhaps not so clear that interaction potentials have no place in relativistic theories.

Moreover, while the use of multi-time wave functions is motivated by relativity, they can be considered also independently of relativity. In particular, one can consider multi-time equations of the form \eqref{phiHj} that are not covariant but make use of a special Lorentz frame and will have a different $H_j$ after transformation to a different frame. So it may again be surprising that even if we do not require Lorentz invariance but only consistency of the multi-time equations, interaction potentials are excluded. As a consequence, in view of our results in this paper and in \cite{pt:2013c,pt:2014a,pt:2014b}, the mere introduction of multi-time wave functions, with space-time configurations as arguments, naturally leads us to considering particle creation and annihilation.

On the other hand, further possibilities are known to exist, at least mathematically, and at least if we are willing to consider $H_j$ that are nonlocal in time: Several authors \cite{FL70,DV71,DV81,VAC97,Hah09} have given examples of Lorentz-invariant equations of the form \eqref{phiHj} for $N=2$ particles that are, at least in some sense, consistent. In these examples, the interaction terms neither are multiplication operators (like potentials) nor involve particle creation; instead, $H_j\phi$ involves integrating $\phi$ over some time interval (partly in the future), so that these equations may not determine $\phi$ from initial data as in a Cauchy problem.

As a last remark, our result may seem unsurprising in view of the result of Currie, Jordan, and Sudarshan \cite{CJS63}, who showed for $N=2$ particles that classical mechanics, in a particular Hamiltonian formulation, cannot be made relativistic except in the absence of interaction. Then again, it is not clear whether and why that particular Hamiltonian formulation should be regarded as the appropriate classical analog of the framework of multi-time Schr\"odinger equations; see \cite{pt:2013e} for further discussion.

\subsection{Example of Inconsistent Multi-Time Equations}

Let us now look at an explicit example of multi-time equations with potentials. Consider the two-particle Hilbert space $\Hilbert=L^2(\RRR^3\times \RRR^3,\CCC^d)$ and the Hamiltonian
\be
H = H^\free + V
= H_1^\free + H_2^\free + \frac{1}{\|\vx_1-\vx_2\|}
\ee
in the usual single-time picture, where the $H_j^\free$ are the free Schr\"o\-din\-ger ($d=1$) or free Dirac ($d=16$) Hamiltonians. To begin with, in view of \eqref{HjH}, how should the potential be distributed on the two partial Hamiltonians? If $V(\vx_1,\vx_2)$ were of the form $V_1(\vx_1)+V_2(\vx_2)$ then it would be natural to define $H_j=H_j^\free+V_j$, but such a potential $V$ would represent an external field and not interaction between the two particles. One obvious (though perhaps unnatural) possibility for the Coulomb potential is to attribute half of the potential to each of the partial Hamiltonians, i.e., to set
\be
H_j = H_j^\free + \frac{1}{2\|\vx_1-\vx_2\|}.
\ee
But then the consistency condition \eqref{consistency} is violated. Indeed, if the $H_j^\free$ are the free Schr\"o\-din\-ger Hamiltonians,
\be
H_j^\free=-\frac{1}{2m}\Laplace_j\,,
\ee
then
\be\label{H1H2LaplaceCoulomb}
[H_1,H_2] = \frac{\vx_1-\vx_2}{2m\|\vx_1-\vx_2\|^3} \cdot (\nabla_1+\nabla_2) \neq 0,
\ee
and if the $H_j^\free$ are the free Dirac Hamiltonians (we set $c=1$),
\be
H_j^\free=-i\vect{\alpha}_j\cdot\nabla_j + \beta_j m\,,
\ee
then
\be\label{H1H2DiracCoulomb}
[H_1,H_2] = \frac{i(\vx_1-\vx_2)}{2\|\vx_1-\vx_2\|^3} \cdot (\vect{\alpha}_1+\vect{\alpha}_2) \neq 0.
\ee
(Here, $\valpha_j$ means the 3-vector consisting of the three Dirac alpha matrices, acting on the spin index of particle $j$.) As above, for time-independent Hamiltonians the consistency condition \eqref{consistency} amounts to $[H_1,H_2]=0$. Thus, the most obvious choice of multi-time equations for two particles with a Coulomb potential is inconsistent.

(Furthermore, the only joint solution $\phi:\RRR^2\to\Hilbert$ of \eqref{phiHj} is zero. Indeed, an initial wave function which is mapped to zero by the operator in \eqref{H1H2LaplaceCoulomb} must be constant along the line $\{(\va+(1+s)\vb,\va+s\vb):s\in\RRR\}$ in $\RRR^3\times\RRR^3$ for any $\va,\vb\in\RRR^3$, and thus cannot be square-integrable unless it vanishes almost everywhere. The kernel of the operator in \eqref{H1H2DiracCoulomb} consists of those wave functions $\psi:\RRR^6\to\CCC^{16}$ such that, at almost every $(\vx_1,\vx_2)$, $\psi(\vx_1,\vx_2)$ is an eigenvector of $(\vx_1-\vx_2)\cdot \valpha_1$ with eigenvalue $\pm \|\vx_1-\vx_2\|$ and simultaneously an eigenvector of $(\vx_1-\vx_2)\cdot \valpha_2$ with eigenvalue $\mp \|\vx_1-\vx_2\|$; non-zero elements of the kernel will not remain in the kernel under $\exp(-iH_1t)$ or $\exp(-iH_2t)$.)

Several previous authors, starting from Bloch \cite{bloch:1934}, were aware of this kind of difficulties with potentials; some \cite{FL70,Marx:1974} mentioned it explicitly.

\subsection{Brief Overview of Results}

\subsubsection{Results on Consistency Condition}

The first type of results that we present (Theorems~\ref{thm:SchrN} and \ref{thm:SchrN_t}) are precise formulations and proofs of the statement that the condition \eqref{consistency} is necessary and sufficient for the consistency of the system of multi-time equations \eqref{phiHj}. In this context we also elucidate (Section~\ref{sec:pathindependence}) a perspective from which \eqref{consistency} expresses the vanishing of the curvature of a gauge connection on a vector bundle with fibers $\Hilbert$ over the space $\RRR^N$ spanned by the time axes, such that time evolution corresponds to parallel transport in this bundle.

\subsubsection{Inconsistency Results}

The main goal of this paper is to prove in great generality that interaction potentials in multi-time equations lead to violations of the consistency condition \eqref{consistency}. Our results (Theorems~\ref{thm:inconsistentV1}--\ref{thm:inconsistent2ndorder}) cover arbitrary smooth potential functions, also time-dependent ones, such as $V_j(t_1,\vx_1,\ldots,t_N,\vx_N)$; they cover the free Dirac Hamiltonian and the Laplace operator, and in fact all self-adjoint differential operators up to second order, as free Hamiltonians. 

Furthermore, writing $x=(x^0,x^1,x^2,x^3)=(x^0,\vx)=(t,\vx)$ for a space-time point, it is reasonable to demand of $\phi$ only that it be defined on the set of \emph{spacelike configurations} of $N$ particles,
\be\label{sSdef}
\sS=\Bigl\{(x_1,\ldots,x_N)\in(\RRR^{4})^N: \:\:\forall j\neq k:x_j \sim x_k \text{ or } x_j = x_k\Bigr\}
\ee
(where $x\sim y$ means that $x$ is spacelike to $y$),\footnote{We say that $x$ is spacelike to $y$ if and only if $(x^0-y^0)^2-\|\vx-\vy\|^2<0$. For a \emph{spacelike configuration}, we allow repeated entries, $x_j=x_k$.}
and not on all of $(\RRR^4)^N$.\footnote{Bloch \cite{bloch:1934} has argued first that multi-time wave functions $\phi$ should be defined only on $\sS$, a view that we share. As we show in \cite{pt:2013c}, it is actually the case that interaction implemented by the creation and annihilation of particles is consistent on spacelike configurations, but not on all space-time configurations (i.e., consistent on the analog of $\sS$ for a variable number of particles, but not on that of $(\RRR^4)^N$).} In that case, also the partial Hamiltonians need to be defined only on $\sS$ (e.g., as differential operators), and also the consistency condition can only be expected to hold on $\sS$. Among our inconsistency results, we also prove that even if the consistency condition holds only on $\sS$, all consistent potentials are interaction-free.

\subsubsection{Consistency with Cut-Off Length $\delta$}

Another result we prove (Theorem~\ref{thm:delta_model}) concerns a scenario with a cut-off length $\delta>0$. Although a cut-off length is usually not considered for potentials but only for particle creation, we consider it here for potentials because that allows for a consistent kind of multi-time equations. In an arbitrary but fixed Lorentz frame $L$, we consider instead of $\sS$ the set of \emph{$\delta$-spacelike configurations},
\begin{equation}\label{sSdeltadef}
\sS_{\delta} = \Bigl\{ (x_1,\ldots,x_N) \in (\RRR^4)^N:\:\: \forall j \neq k: x^0_j = x^0_k ~\text{or}~ \|\vx_j - \vx_k\| > |x^0_j - x^0_k| + \delta \Bigr\}\,.
\end{equation}
This is the set of those spacelike configurations in which additionally for each pair of particles either the times are equal or the spatial distance from one particle's light cone to the other particle is bigger than $\delta$. Figure~\ref{figure:std_partition} shows an example of such a configuration.

\begin{figure}[htbp]
\centering
\includegraphics[width=350pt,keepaspectratio]{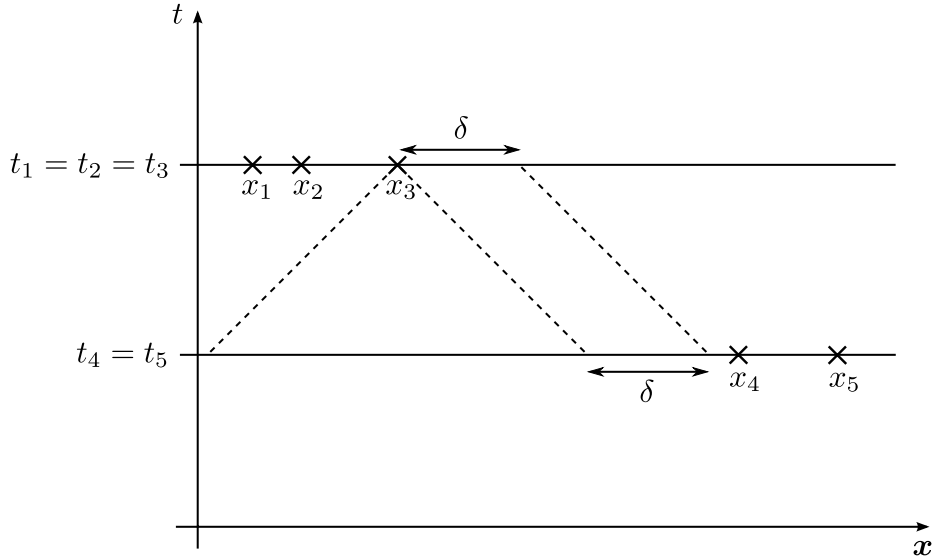}
\centering
\caption{\small{A space-time configuration belonging to $\sS_{\delta}$}}
\label{figure:std_partition}
\end{figure}

Our result, Theorem~\ref{thm:delta_model}, asserts that for a pair potential of range $\delta$ or less, the obvious multi-time equations on $\sS_\delta$ (with Dirac operators as $H_j^\free$) are consistent.\footnote{As a consequence, the statement made in the title of this paper is not true on $\sS_\delta$, and the title should be understood as meaning that multi-time Schr\"o\-din\-ger equations without cut-off length cannot contain interaction potentials.}

In more detail, for any $\delta$-spacelike configuration the $N$ particles can be grouped together in families that have equal time coordinate. The obvious multi-time equations say that particles belonging to different families do not interact (as they have distances greater than $\delta$ anyway), while each family, having only one joint time coordinate, satisfies one Schr\"o\-din\-ger equation that contains interaction potentials. It is therefore not surprising that the multi-time equations, one equation per family, are consistent. Yet, the proof requires some work, as the grouping into families is not fixed but varies over $\sS_\delta$. (Note also that two particles at a distance $>\delta$ can interact by interacting with a third particle, travelling from one to the other if necessary.) 

This model breaks Lorentz invariance in two ways: the set $\sS_\delta$ is not Lorentz invariant, and the time evolution involves $L$-instantaneous interaction within each family (over distances up to $\delta$), and thus superluminal signaling. This leads to the question, which we address now, whether a fully relativistic version of the example could be obtained by letting $\delta\to 0$.

\subsubsection{On the Limit $\delta\to0$}

The answer is negative: There is no consistent set of multi-time equations with interaction of range 0 based on the Dirac equation. Indeed, such a multi-time evolution would involve a wave function $\phi$ defined on the set of spacelike configurations of $N$ particles, which at each spacelike configuration $(x_1,\ldots,x_N)$ (away from the set where $x_j=x_k$) would satisfy the free multi-time Dirac equations,
\be\label{freeDiracmultitime}
i\frac{\partial\phi}{\partial t_j} = \bigl(-i\valpha_j \cdot \nabla_j + \beta_jm\bigr)\phi
\ee
for $j=1,\ldots,N$. While these equations look like the free equations, they alone may not completely determine the time evolution, as the set $\sS$ of spacelike configurations has a non-empty boundary $\partial \sS$; to determine the time evolution it might be necessary to specify boundary conditions on $\partial \sS$, and one might hope that a suitable choice of boundary condition will define a consistent interacting multi-time evolution. However,  consider (in any one Lorentz frame $L$) the 1-time wave function $\psi$ obtained from the multi-time wave function $\phi$ as in \eqref{phipsi}. As a consequence of \eqref{freeDiracmultitime}, $\psi$ would satisfy the free $N$-particle 1-time Dirac equation
\be\label{freeDirac1time}
i\frac{\partial \psi}{\partial t} = \sum_{j=1}^N \Bigl( -i\valpha_j\cdot \nabla_j + \beta_j m_j\Bigr) \psi
\ee
on the set 
\be
\RRR^{3,N}_{\neq} = \Bigl\{ (\vx_1,\ldots,\vx_N)\in(\RRR^3)^N: \:\: \forall j\neq k :\vx_j \neq \vx_k\Bigr\}
=(\RRR^3)^N\setminus D
\ee
of configurations without collisions. It is known \cite{svendsen} that the Dirac equation does not allow for point interactions if the dimension of physical space is 3; that is, if \eqref{freeDirac1time} holds away from the diagonal $D$ then it also holds on the diagonal $D$; put yet differently, the free Dirac Hamiltonian as in \eqref{freeDirac1time} is essentially self-adjoint on $C_0^\infty\bigl(\RRR^{3N}\setminus D,(\CCC^4)^{\otimes N}\bigr)$, and its unique self-adjoint extension is again a free Dirac operator. Thus, the theory in the limit $\delta\to 0$ is free of interaction.

\section{The Consistency Condition}\label{sec:consistency}

Since the time derivatives $\partial/\partial t_j$ and $\partial/\partial t_k$ commute, \eqref{consistency} is just a more compact way of writing 
\be\label{consistency2}
\bigl[H_j,H_k \bigr] - i\frac{\partial H_k}{\partial t_j} + i \frac{\partial H_j}{\partial t_k} =0 \quad \forall j\neq k\,.
\ee
We now want to convey why this condition is necessary and sufficient for the multi-time equations \eqref{phiHj} to possess a joint solution for every initial condition. We will approach this point from several angles.

\subsection{Heuristic Derivation}

First, regard the multi-time equations \eqref{phiHj} as $N$ equations about a function $\phi: \RRR^{4N} \to S$ (with $S$ the spin-space, e.g., $S=\CCC$ for spinless particles, or $S=(\CCC^{4})^{\otimes N}$ for Dirac particles). If $\phi$ is a joint solution then
\be\label{consistency_phi}
\biggl[ i\frac{\partial}{\partial t_j}-H_j , i\frac{\partial}{\partial t_k}-H_k \biggr] \phi = 0\quad \forall j\neq k\,. 
\ee
Now suppose that arbitrary times $\tau_1,\ldots,\tau_N$ can be chosen as initial times, at which arbitrary initial data $\phi|_{t_1=\tau_1,\ldots,t_N=\tau_N}$ can be specified. Then the left-hand side of \eqref{consistency_phi} has to vanish at $t_1=\tau_1,\ldots,t_N=\tau_N$ for arbitrary $\phi|_{t_1=\tau_1,\ldots,t_N=\tau_N}$, and therefore the commutator has to vanish at $t_1=\tau_1,\ldots,t_N=\tau_N$. Since the $\tau_j$ were arbitrary, the commutator has to vanish everywhere, which is what we wanted to derive.

It is sometimes important to note that, for the condition \eqref{consistency} to characterize consistency, the commutator has to vanish not only on all \emph{joint solutions} $\phi$ of \eqref{phiHj}, but on \emph{all possible initial conditions}, or, equivalently (since, as visible from \eqref{consistency2}, the left-hand side of \eqref{consistency_phi} does not involve the time derivative of $\phi$), on \emph{all functions} $\RRR^{4N}\to S$.

In many relevant cases, the commutator can also be considered locally (i.e., at a single point in $\RRR^{4N}$), e.g., when the $H_j$ are differential operators; then the commutator may vanish on some subset $U$ of $\RRR^{4N}$ but not on its complement. This case is relevant in connection with the fact that the Dirac equation has a finite propagation speed (given by the speed of light, which we have set to 1), so that a solution $\psi(t,\vx)$ of the (1-particle) Dirac equation depends on the initial data $\psi(0,\cdot)$ only through the initial data in the closed ball $\overline B_{|t|}(\vx)$ around $\vx$ with radius $|t|$ (called the domain of dependence). If every $H_j$ has propagation speed 1, then $\phi(\tau_1+\varepsilon,\vx_1,\ldots,\tau_N+\varepsilon,\vx_N)$ depends on initial data at times $\tau_1,\ldots,\tau_N$ only on $\overline B_\varepsilon(\vx_1)\times \cdots \times \overline B_\varepsilon(\vx_N)$, and for solving \eqref{phiHj} on the Cartesian product of the cones between $(\tau_j+\varepsilon,\vx_j)$ and $(\tau_j, \overline B_\varepsilon(\vx_j))$ we only need that this product set is contained in $U$. In particular, for solving \eqref{phiHj} on $\sS$ we only need that \eqref{consistency} holds on $\sS$.

\subsection{Exact Formulations}

Another way of thinking about the consistency condition is from a Hilbert space perspective. Regard $\phi$ as a function $\phi:\RRR^N \to \Hilbert$ on the space $\RRR^N$ spanned by the $N$ time axes, with values in a Hilbert space, e.g., $\Hilbert=L^2(\RRR^{3N},S)$. The partial Hamiltonians are $N$ operator-valued functions $H_1(t_1,\ldots,t_N),\ldots,H_N(t_1,\ldots,t_N)$.

A simple case is that in which the $H_j$ are time-independent:

\begin{thm}\label{thm:SchrN}
Let $\Hilbert$ be a Hilbert space, and let $H_1,\ldots,H_N$ be (time-independent) self-adjoint operators in $\Hilbert$. Then the system of equations \eqref{phiHj} possesses a strong solution for every initial condition $\phi(0,\ldots,0)\in\Hilbert$ if and only if the consistency condition $[H_j,H_k]=0$ holds (in the spectral sense) for all $j \neq k \in\{1,\ldots,N\}$.
\end{thm}

The proof is given in Section~\ref{sec:proofs_consistency}.
Here, a \emph{strong solution} means a function $\phi:\RRR^N\to\Hilbert$ such that
\be\label{strongdef}
\phi(t_1,\ldots,t_N) = e^{-i H_j t_j} \phi(t_1,\ldots,t_{j-1},0,t_{j+1},\ldots,t_N)
\ee
for every $j=1,\ldots,N$, in analogy to the usual terminology that a strong solution of the single-time Schr\"o\-din\-ger equation $i\partial_t \psi = H\psi$ is a function $\psi:\RRR\to\Hilbert$ such that $\psi(t) = e^{-iHt} \psi(0)$. Furthermore, to say that $[H_1,H_2]=0$ ``in the spectral sense'' means that, for all measurable sets $A_1,A_2\subseteq \RRR$, the corresponding spectral projections $P_1=1_{A_1}(H_1)$ and $P_2=1_{A_2}(H_2)$ commute. In case $H_1$ and $H_2$ are bounded operators (and thus defined on all of $\Hilbert$), this statement is equivalent to $H_1H_2=H_2H_1$. (For unbounded operators, the expression $H_1H_2-H_2H_1$ may not be defined on a dense domain because the range of $H_1$ may not be contained in, or even may not overlap non-trivially with, the domain of $H_2$ and vice versa.)

Let us return to the general case $H_j=H_j(t_1,\ldots,t_N)$.

\begin{thm}\label{thm:SchrN_t}
Let $\Hilbert$ be a Hilbert space, and let $H_1,\ldots,H_N$ be smooth functions on $\RRR^N$ with values in the bounded operators on $\Hilbert$. Then the system of equations \eqref{phiHj} possesses a solution $\phi:\RRR^N\to\Hilbert$ for every initial condition $\phi(0,\ldots,0)\in\Hilbert$ if and only if the consistency condition \eqref{consistency} holds.
\end{thm}

The derivative $\partial/\partial t_j$ in \eqref{phiHj} is understood here as the limit in the Hilbert space topology of the appropriate difference quotient. The proof of Theorem~\ref{thm:SchrN_t} is also given in Section~\ref{sec:proofs_consistency}. The restriction to \emph{bounded} $H_j(t_1,\ldots,t_N)$ in Theorem~\ref{thm:SchrN_t} can presumably be relaxed if a more refined proof is used.
For finite-dimensional $\Hilbert$, Theorem~\ref{thm:SchrN_t} was already known \cite{GU12}.

\subsection{Path Independence}\label{sec:pathindependence}

Let us explore further the view of $\phi$ as an $\Hilbert$-valued function on $\RRR^N$. Already the simplified derivation of the consistency condition in Section~\ref{sec:consistency1} has made clear that consistency is related to a certain type of path independence: It was relevant to consistency that we could either first increase the $t_1$ variable from its initial to its final value, and then the $t_2$ variable, or vice versa. The first way of obtaining $\phi(t_1,t_2)$ from $\phi(0,0)$ proceeds along a path that is the polygonal chain from $(0,0)$ to $(t_1,0)$ to $(t_1,t_2)$; the second along the path from $(0,0)$ to $(0,t_2)$ to $(t_1,t_2)$; see Figure~\ref{fig:twopaths}. 

\begin{figure}[ht]
\centering
\includegraphics[width=250pt,keepaspectratio]{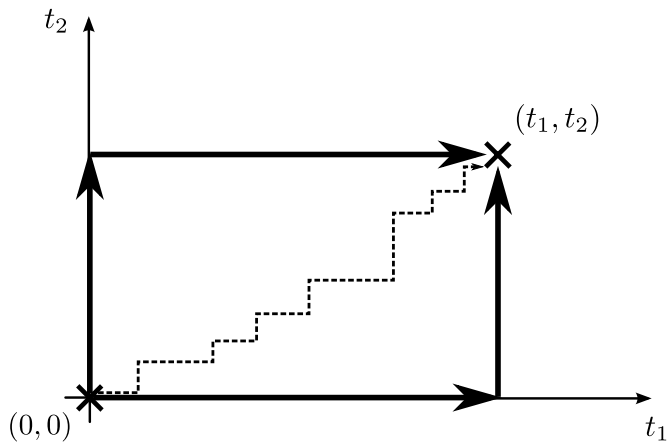}
\centering
\caption{\small{Two paths in the $t_1t_2$-plane from $(0,0)$ to $(t_1,t_2)$ associated with two ways of computing $\phi(t_1,t_2)$ from $\phi(0,0)$: either first increase the $t_1$ variable and then $t_2$, or first increase $t_2$ and then $t_1$. Dashed: another path from $(0,0)$ to $(t_1,t_2)$ associated with another way, first increase the $t_1$ variable a bit but not to the desired final value, then increase the $t_2$ variable a bit, then $t_1$ again etc..}}
\label{fig:twopaths}
\end{figure}

As further alternatives, we could first increase $t_1$ a bit, then $t_2$, then $t_1$ some more, then $t_2$, etc.. More generally, also for $N$ time variables, with every path $\gamma:[0,1]\to\RRR^N$ from the initial point $\gamma(0)=\art^{\,i}=(t_1^i,\ldots, t_N^i)$ to the final point $\gamma(1)=\art^f=(t_1^f,\ldots,t_N^f)$, there is associated an evolution operator 
\be
U_\gamma = \mathcal{T} e^{-i\int_\gamma \sum_j H_j \, dt_j}\,,
\ee
which means the path-ordered exponential integral, i.e., the value $U(1)$ of the solution of the differential equation 
\be
\frac{dU(s)}{ds}=-i \sum_j H_j(\gamma(s)) \, \frac{d\gamma_j(s)}{ds} \, U(s)
\ee
with initial condition $U(0)=I$, the identity operator; for bounded $H_j$, it is given by the Dyson series (see Section~\ref{sec:proofs_consistency} for more detail). 

The consistency of \eqref{phiHj} is then equivalent to saying that every path $\gamma$ from $\art^{\,i}$ to $\art^f$ yields the same operator $U_\gamma$; or, put differently, since any point in $\RRR^N$ could serve as $\art^{\,i}$ or $\art^f$, that $U_\gamma$ depends only on the endpoints of $\gamma$.

We note that, as discussed in more detail in \cite{pt:2013c}, a similar situation occurs for the Tomonaga--Schwinger equation, which defines the evolution of a wave function in Hilbert space along any path in the set of spacelike hypersurfaces; if the consistency condition of the Tomonaga--Schwinger equation is satisfied, then the evolution is path-independent, and thus depends only on the initial and the final hypersurface but not on the foliation used to interpolate between the two.

Returning to the consistency of \eqref{phiHj}, the 
evolution operator $U_\gamma$ and its path-independence can be naturally expressed in terms of a gauge connection of a vector bundle. Consider as the base manifold the space $\RRR^N$ spanned by the $N$ time axes, take as the fiber spaces copies of $\Hilbert$ (so it is a trivial vector bundle), and regard $\phi$ as a cross-section of this vector bundle. Define a gauge connection (or ``covariant derivative'') on this bundle by\footnote{We use the symbol $\nabla_j$ in this paper in two different meanings: here, it means the covariant derivative in the $t_j$-direction, whereas in some other places including \eqref{H1H2LaplaceCoulomb}, \eqref{freeDiracmultitime}, and \eqref{freeDirac1time}, it means the gradient $(\partial/\partial x_{j,1},\partial/\partial x_{j,2},\partial/\partial x_{j,3})$ with respect to spatial variables. It should always be clear from the context which meaning is intended.}
\be\label{covder}
\nabla_j = \partial_j -iA_j
\ee
with $\partial_j=\partial/\partial t_j$ and connection coefficients $A_j=-H_j$.\footnote{Put differently, a gauge connection can be described, relative to some other gauge connection, by a one-form with values in the Lie algebra of the gauge group; here, the reference connection [$\partial$ in \eqref{covder}] is the one along which the Hilbert spaces in the bundle are identified with $\Hilbert$ (and thus with each other), the Lie algebra of the gauge group consists of operators on $\Hilbert$, and a one-form on $\RRR^N$ can be specified by specifying its $N$ components $A_j=-H_j$.} Then $U_\gamma$ coincides with the parallel transport operator along $\gamma$ associated with the gauge connection, a joint solution $\phi$ is a cross-section for which all covariant derivatives vanish, and path-independence is equivalent to saying that all closed curves $\gamma$ have trivial holonomy (i.e., $U_\gamma=I$). By the non-Abelian Stokes theorem (see the proof of Theorem~\ref{thm:SchrN_t} in Section~\ref{sec:proofs_consistency}), the holonomy of a closed curve $\gamma$ equals the ordered exponential integral of the curvature over any oriented 2-surface $\Sigma$ whose oriented boundary is $\gamma$. As a consequence, a gauge connection has only trivial holonomies if and only if its curvature $F$ vanishes. Here, $F$ is an operator-valued 2-form; by the standard formula for computing curvature from the connection coefficients, the components of $F$ are 
\be\label{Fjkdef}
F_{jk} = -\frac{\partial H_k}{\partial t_j} +\frac{\partial H_j}{\partial t_k} -i  \bigl[ H_j,H_k \bigr]\,.
\ee
Thus, path-independence (and thus consistency) is equivalent to
\be
F_{jk} =0 \quad \forall j\neq k\,,
\ee
which coincides with condition \eqref{consistency} in the form \eqref{consistency2}. We have thus obtained another derivation of the equivalence between \eqref{consistency} and consistency; a precise version of this derivation is formulated in Section~\ref{sec:proofs_consistency}.

\subsection{Comparison with the Frobenius Theorem}

The statement that \eqref{consistency} is necessary and sufficient for the consistency of a system of multi-time equations bears some similarity with the Frobenius theorem of differential topology, which concerns the following. Suppose that with every point $q\in\RRR^d$ is associated a subspace $S_q\subset \RRR^d$ of dimension $N<d$; the family of the subspaces $S_q$ is called \emph{integrable} if and only if there exists an $N$-dimensional foliation of $\RRR^d$ such that $S_q$ is the tangent space of the foliation at $q$. The Frobenius theorem provides a necessary and sufficient condition for integrability. Equivalently, the situation can be expressed in terms of a system of $N$ partial differential equations for a function $\phi:\RRR^d\to\RRR^{d-N}$, $\phi(q)=\phi(q_1,\ldots,q_d)$, of the form
\be\label{FrobeniusPDE}
\sum_{n=1}^d f_{kn}(q)\,  \frac{\partial \phi}{\partial q_n} =0\quad \forall k=1,\ldots,N
\ee
with real-valued coefficients $f_{kn}$.

Here, the leaves of the foliation are supposed to be the surfaces of constant $\phi=(\phi_1,\ldots,\phi_{d-N})$; they do form a foliation if the matrix $D(q)=(D_{in}) = (\partial \phi_i/\partial q_n)$ has full rank $d-N$ at every $q\in\RRR^d$. At every $q$, the $N$ vectors (written as directional derivative operators)
\be\label{FrobLdef}
L_k = \sum_{n=1}^d f_{kn}(q)\,  \frac{\partial}{\partial q_n}
\ee
span the subspace $S_q$; Eq.~\eqref{FrobeniusPDE} expresses that the foliation is tangent to $S_q$. The family $S_q$ is integrable if and only if there is a solution $\phi$ of \eqref{FrobeniusPDE} such that $D(q)$ has full rank at every $q$; in this case, Eq.~\eqref{FrobeniusPDE} is said to be integrable. The Frobenius theorem states that \eqref{FrobeniusPDE} is integrable in a neighborhood of $q_0\in\RRR^d$ if and only if the operators $L_k$ satisfy the commutator condition
\be
[L_i,L_j] = \sum_{k=1}^N c_{ijk}(q) \, L_k
\ee
in a neighborhood of $q_0$ for suitable functions $c_{ijk}$.

To compare the Frobenius theorem to Theorem~\ref{thm:SchrN_t}, we consider $d=4N$, $q=(x_1,\ldots,x_N)$, and
\be\label{LH}
L_k=\frac{\partial}{\partial x^0_k} +iH_k\,.
\ee
The similarities are that both theorems concern the possibility of a joint solution $\phi$ of several PDEs; that this possibility occurs if and only if the PDEs satisfy a certain integrability condition; and that this condition can be expressed in terms of the commutator $[L_i,L_j]$. The differences are that $\phi$ in Theorem~\ref{thm:SchrN_t} can have any number of components, not just $d-N$;  that Theorem~\ref{thm:SchrN_t} does not guarantee that $D(q)$ has full rank at every $q$; relatedly, that a joint solution $\phi$ of \eqref{phiHj} does not, in general, define an $N$-dimensional foliation of $\RRR^d$; that the operators $H_k$ in Theorem~\ref{thm:SchrN_t} can have matrix-valued coefficients, do not have to be of first order, and do not have to be differential operators;\footnote{Strictly speaking, Theorem~\ref{thm:SchrN_t} does not apply to differential operators because they are unbounded; but an appropriate version of Theorem~\ref{thm:SchrN_t} will be true of differential operators.} and that the commutator in \eqref{consistency} has to vanish, rather than being a linear combination of the $L_k$. Thus, neither theorem is a special case of the other. 

A statement of the following type can be regarded as a generalized Frobenius theorem that covers both the original Frobenius theorem and the consistency of multi-time Schr\"odinger equations as special cases: 
\begin{itemize}
\item {\it Let $\Hilbert$ be a vector space. The system of partial differential equations for $\phi:\RRR^N\to \Hilbert$,
\be\label{nonlinear}
\frac{\partial \phi}{\partial t_j}= f_j\Bigl(t_1,\ldots,t_N,\phi(t_1,\ldots, t_N)\Bigr)\,,
\ee
possesses a solution $\phi$ on $\RRR^N$ for every initial condition $\phi(0,\ldots,0)=\phi_0$ if and only if the functions $f_j:\RRR^N\times \Hilbert \to \Hilbert$, $j=1,\ldots,N$, everywhere satisfy the consistency condition
\be\label{consistency-nonlinear}
\frac{\partial f_j}{\partial t_k}+ f_k \cdot \nabla_\phi f_j=
\frac{\partial f_k}{\partial t_j}+ f_j \cdot \nabla_\phi f_k
\ee
for all $k=1,\ldots,N$. In this case, the solution is unique.}
\end{itemize}
For finite-dimensional $\Hilbert$ (and under  certain technical assumptions on the $f_j$), this statement can be found as Theorems 2.1 and 2.2 in \cite{GU12}. On the other hand, for $f_j(t_1,\ldots,t_N,\phi)=-iH_j(t_1,\ldots,t_N)\phi$, \eqref{nonlinear} reduces to \eqref{phiHj} and \eqref{consistency-nonlinear} to \eqref{consistency}, and Theorems~\ref{thm:SchrN} and \ref{thm:SchrN_t} above provide precise versions of the statement. The original Frobenius theorem (or at least a variant thereof) is a special case of the statement for $(q_1,\ldots,q_d)=(t_1,\ldots,t_N,x_1,\ldots,x_{d-N})$, $L_k=\partial/\partial t_k + \sum_{n=1}^{d-N} f_{kn}(q)\, \partial /\partial x_n$, $\Hilbert$ a suitable space of functions of $(x_1,\ldots,x_{d-N})$, and $f_j(t_1,\ldots,t_N,\phi)=- \sum_{n=1}^{d-N} f_{kn}(q) \, \partial \phi/\partial x_n$.

\section{Results on Inconsistency of Potentials}\label{sec:potentials}

We now state our results about the inconsistency of interaction potentials. The proofs are postponed to Section~\ref{sec:proofs_inconsistentV}. The results cover, as the free Hamiltonians, both the free Schr\"o\-din\-ger Hamiltonian\footnote{There should be no difficulty with distinguishing when the symbol $i$ denotes the unit imaginary number from when it denotes a particle label.}
\be\label{freeSchr}
H_i^\free=-\frac{1}{2m}\Laplace_i ,
\ee
and the free Dirac Hamiltonian
\be\label{freeDirac}
H_i^\free=-i\vect{\alpha}_i\cdot\nabla_i + \beta_i m .
\ee
Since we are more interested in the Dirac case, we formulate our results first for first-order operators (Theorems~\ref{thm:inconsistentV1}--\ref{thm:inconsistentV4}), and present the corresponding results for second-order operators afterwards in Theorem~\ref{thm:inconsistent2ndorder}.

So consider $\phi_{s_1,\ldots, s_N}(x_1,\ldots, x_N)$ with $s_j=1,\ldots,4$ the spin index of the $j$-th particle; that is, $\phi:\RRR^{4N}\to(\CCC^4)^{\otimes N}$. In the following Theorem~\ref{thm:inconsistentV1} we consider real scalar interaction potentials.

\begin{thm}\label{thm:inconsistentV1}
Suppose $H_i=H_i^\mathrm{free}+V_i(x_1,\ldots,x_N)$, where $H_i^{\mathrm{free}}$ is the free Dirac Hamiltonian \eqref{freeDirac} acting on $\vx_i$ (and the $i$-th spin index), and $V_i:\RRR^{4N}\to \RRR$ is smooth.
The consistency condition \eqref{consistency} is satisfied on $\RRR^{4N}$ only if the evolution \eqref{phiHj} is gauge-equivalent to a non-interacting one, i.e., there are smooth real-valued functions $\theta(x_1,\ldots,x_N)$ and, for every $i\in\{1,\ldots,N\}$, $\tilde{V}_i(x_i)$ such that
\be\label{tildephidef}
\tilde\phi(x_1,\ldots,x_N) := e^{i\theta(x_1,\ldots,x_N)}\phi(x_1,\ldots,x_N)
\ee
satisfies the equations
\be\label{inconVfree}
i\frac{\partial\tilde\phi}{\partial t_i} = \Bigl(H_i^\mathrm{free}+\tilde{V}_i(x_i)\Bigr)\tilde\phi
\ee
for $i=1,\ldots,N$.
\end{thm}

Theorem~\ref{thm:inconsistentV1} can easily be generalized to more arbitrary first-order differential operators.

\begin{thm}\label{thm:inconsistentV2}
Theorem~\ref{thm:inconsistentV1} is still true if the free Dirac operator is replaced by any first-order differential operator
\be\label{firstorderoperator}
H_i^\free = -i\sum_{a=1}^3 A_{i,a}\frac{\partial}{\partial x_{i,a}} + B_i
\ee
for every $i=1,\ldots,N$, where the coefficients $A_{i,a}$ and $B_i$ are self-adjoint $k_i \times k_i$ matrices acting on the index $s_i=1,\ldots,k_i$ of the wave function (referring to the spin space $\CCC^{k_i}$ of the $i$-th particle), and for each $i$ the four matrices $A_{i,1},A_{i,2},A_{i,3},I$ (with $I$ the identity matrix) are linearly independent.\footnote{This is the case for the free Dirac Hamiltonian with $k_i=4$, $A_{i,a} = \alpha_a$ acting on $s_i$ and $B_i=\beta m$ acting on $s_i$; indeed, the three $\alpha$ matrices and $I$ are linearly independent in the space of self-adjoint $4\times 4$ matrices.} That is, let, for every $i$,
\be
H_i=H_i^\free+V_i(x_1,\ldots,x_N)
\ee
with arbitrary real-valued, smooth potential functions $V_i$.
The consistency condition \eqref{consistency} is satisfied only if the multi-time evolution \eqref{phiHj} defined by $H_1,\ldots,H_N$ is gauge-equivalent to a non-interacting one, i.e., (as before) there are smooth real-valued functions $\theta(x_1,\ldots,x_N)$ and $\tilde{V}_i(x_i)$ such that $\tilde\phi$ given by \eqref{tildephidef} satisfies \eqref{inconVfree}.
\end{thm}

The assumption that $A_{i,a}$ and $B_i$ be self-adjoint can actually be dropped; it is not used in the proof, nor in that of Theorem~\ref{thm:inconsistentV3}. However, as far as we are aware, it is satisfied in all examples of physical interest. The assumption that $A_{i,a}$ be self-adjoint is needed in Theorem~\ref{thm:inconsistentV4}. 

In the following theorem, we consider spacelike configurations. As remarked in the introduction, it is reasonable to expect that multi-time wave functions are defined only on $\sS$. It is useful to also have a notation for the set of spacelike configurations \emph{without collisions} (i.e., demanding $x_j \neq x_k$),
\be\label{sSneqdef}
\sS_{\neq}=\Bigl\{(x_1,\ldots,x_N)\in(\RRR^{4})^N: \:\:\forall j\neq k:x_j \sim x_k\Bigr\}\,,
\ee
where $x\sim y$ means that $x$ is spacelike to $y$.

\begin{thm}\label{thm:inconsistentV3}
Theorem~\ref{thm:inconsistentV2} is still true if the consistency condition holds only on the set $\sS_{\neq}$ of collision-free spacelike configurations and the multi-time wave function is defined only on $\sS_{\neq}$.
\end{thm}

Theorem~\ref{thm:inconsistentV3} yields that \eqref{inconVfree} holds on $\sS_{\neq}$. In order to see that this actually implies the absence of interaction, note that it implies the corresponding 1-time equation
\be\label{tildepsinonint}
i\frac{\partial \tilde\psi}{\partial t} = \sum_{i=1}^N\bigl(H_i^\free + \tilde{V}_i(t,\vx_i) \Bigl) \tilde\psi
\ee
on $\RRR^{3N}\setminus D$, where $D$ is the set of collision configurations (``the diagonal''). Since the Dirac equation does not allow for point interactions \cite{svendsen} if the dimension of physical space is 3, \eqref{tildepsinonint} holds also on $D$. (More precisely, the only self-adjoint Hamiltonian in $L^2(\RRR^{3N},(\CCC^4)^{\otimes N})$ that agrees with \eqref{tildepsinonint} on smooth functions with compact support away from $D$ is the obvious, non-interacting one.) And a multi-time evolution law can hardly be called interacting if it is non-interacting on the equal-time configurations.\footnote{Even more, one can argue as follows that the only reasonable multi-time evolution obeying \eqref{inconVfree} on $\sS_{\neq}$ is the one obeying \eqref{inconVfree} on $\RRR^{4N}$. Grouping the particles into families with equal time coordinate as in the proof of Theorem~\ref{thm:delta_model} in Section~\ref{sec:proof_delta}, consider the evolution first on the set of configurations with $L=1$ families, then with $L=2$, etc.. As just pointed out, for $L=1$ the only acceptable evolution (i.e., with a self-adjoint Hamiltonian) is non-interacting. Likewise, when considering $L>1$ families, we require the partial Hamiltonian for each family to be self-adjoint, which implies (by the impossibility of point interactions) that it is non-interacting. It then follows that \eqref{inconVfree} holds on all of $\RRR^{4N}$.}

The most general one of our inconsistency theorems is

\begin{thm}\label{thm:inconsistentV4}
Theorems \ref{thm:inconsistentV2} and \ref{thm:inconsistentV3} are still true if the matrices $A_{i,a}$ and $B_i$ are allowed to depend smoothly on $x_i$ (such that still, for each $x_i$, the four matrices $A_{i,1}(x_i)$, $A_{i,2}(x_i)$, $A_{i,3}(x_i)$, $I$ are linearly independent), and if $V_i(x_1,\ldots,x_N)$ is, rather than a real scalar, a self-adjoint $k_i\times k_i$ matrix acting on the index $s_i$; it is understood that also $\tilde{V}_i(x_i)$ is a self-adjoint $k_i\times k_i$ matrix acting on the index $s_i$, while $\theta(x_1,\ldots,x_N)$ is still real.
\end{thm}

The remark after Theorem~\ref{thm:inconsistentV3} applies also to Theorem~\ref{thm:inconsistentV4}.

Theorem~\ref{thm:inconsistentV4} should be regarded as the main result of this paper, and we take it to rule out interaction potentials for covariant multi-time equations. The case of $x_i$-dependent coefficients occurs, for example, when $H_i^\free$ is the free Dirac operator on a curved space-time.

For the proofs of Theorems \ref{thm:inconsistentV3} and \ref{thm:inconsistentV4} we need some auxiliary lemmas about the connectedness of $\sS_{\neq}$ and certain subsets thereof, in particular that $\sS_{\neq} \subset \RRR^{4N}$ is simply connected. These lemmas are stated and proven in Section~\ref{sec:proofs_inconsistentV}, along with the proofs of the theorems from this section.

Finally, the same kind of results can be obtained for Schr\"o\-din\-ger operators instead of Dirac operators:

\begin{thm}\label{thm:inconsistent2ndorder}
Theorems~\ref{thm:inconsistentV1}, \ref{thm:inconsistentV2}, and \ref{thm:inconsistentV4} are also true when $H_i^\free$ is the free Schr\"o\-din\-ger Hamiltonian \eqref{freeSchr} acting on $\vx_i$, or in fact any second-order differential operator,
\be
H_i^\free = \sum_{a,b=1}^3 A_{i,ab}(x_i)\frac{\partial^2}{\partial x_{i,a}\partial x_{i,b}} + \sum_{a=1}^3 B_{i,a}(x_i)\frac{\partial}{\partial x_{i,a}} + C_i(x_i)\,,
\ee
where each of the coefficients $A_{i,ab}(x_i), B_{i,a}(x_i), C_i(x_i)$ is a complex $k_i\times k_i$ matrix acting on the index $s_i$ and depending smoothly on $x_i$, and the $3k_i\times 3k_i$ matrix $\bigl(A_{i,ab}(x_i) \bigr)_{ab}$ has full rank. 

That is, for such $H_i^\free$, for wave functions $\phi:\RRR^{4N}\to \otimes_i \CCC^{k_i}$, and for smooth $V_i(x_1,\ldots,x_N)$ with values in the self-adjoint $k_i\times k_i$ matrices acting on $s_i$, the consistency condition \eqref{consistency} holds on $\RRR^{4N}$ only if  the multi-time evolution \eqref{phiHj} defined by $H_1,\ldots,H_N$ with $H_i=H_i^\free+V_i$ is gauge-equivalent to a non-interacting one, i.e., there are smooth  functions $\theta:\RRR^{4N}\to\RRR$ and $\tilde{V}_i:\RRR^4\to \CCC^{k_i\times k_i}$ such that $\tilde\phi$ given by \eqref{tildephidef} satisfies \eqref{inconVfree} on $\RRR^{4N}$.
\end{thm}

Here, $\CCC^{k\times k}$ denotes the space of complex $k\times k$ matrices.

\section{Result on Consistency of Potentials with Cut-Off Length $\delta$}
\label{sec:deltarange}
\label{sec:delta_model}

In this section we describe a consistent multi-time theory on the set $\sS_\delta$ of $\delta$-spacelike configurations as in \eqref{sSdeltadef} with an interaction pair potential with range $\delta$. We use the  notation $q = (\vect{x}_1, \dots, \vect{x}_N) \in \RRR^{3N}$ with each $\vect{x}_k \in \RRR^3$ and similarly $q^4 = (x_1, \dots, x_N) \in \RRR^{4N}$ with $x_k \in \RRR^4$.

Let us first turn to the way in which the particles are grouped into families. 
Any configuration $q^4\in\sS_\delta$ defines a partition $P_{q^4}$ of the set $\{1,\ldots,N\}$ of all particles by the equivalence classes of the relation that holds between $j$ and $k$ if and only if $t_j=t_k$. A \emph{partition} $P$ of $\{1,\ldots,N\}$ is a set $P = \{ S_1,\dots,S_L \}$ of non-empty subsets $S_\alpha$ of $\{1,\ldots,N\}$ with $\cup_{\alpha=1}^L S_{\alpha} = \{1,\ldots, N\}$ and $S_{\alpha} \cap S_{\beta} = \emptyset$ for $\alpha \neq \beta$. For every partition $P$ we define
\begin{align}
\sS_{\delta,P} = \Bigl\{ q^4 \in \RRR^{4N}: ~&(1)~ \forall \alpha = 1,\dots,L~\forall i,j \in S_\alpha : t_i = t_j \nonumber \\ &(2)~ \forall \alpha \neq \beta~\forall i \in S_\alpha, j \in S_{\beta}: \|\vx_i-\vx_j\| > |t_i-t_j| + \delta \Bigr\}.
\end{align}
Figure~\ref{figure:std_partition} shows a configuration in  $\sS_{\delta,\{S_1,S_2\}}$ with $S_1 = \{ 1,2,3 \}$ and $S_2 = \{ 4,5 \}$. 
Let
\be
\mathcal{P}_N = \Bigl\{ \text{partitions}~ P ~\text{of}~ \{1\ldots N\} \Bigr\}
\ee
and note that\footnote{For most $q^4\in\sS_{\delta,P}$, $P_{q^4}=P$, but for some $q^4\in\sS_{\delta,P}$, $P_{q^4}$ is coarser than $P$; viz., for those $q^4$ for which particles from $S_\alpha$ happen to have the same time coordinate as particles from $S_\beta$, $\beta\neq\alpha$. As a consequence, $\sS_{\delta,P'}$ need not be disjoint from $\sS_{\delta,P}$ for $P'\neq P$. In fact, $P_{q^4}$ is the coarsest partition $P$ such that $q^4\in\sS_{\delta,P}$. For any $q^4\in\sS_{\delta}$, let $FP_{q^4}$ be the partition of $\{1,\ldots,N\}$ formed by the equivalence classes of the transitive hull of the relation that holds between $j$ and $k$ if and only if $\|\vx_j-\vx_k\|\leq |t_j-t_k|+\delta$. Then $FP_{q^4}$ is the finest partition $P$ such that $q^4\in\sS_{\delta,P}$. Moreover, $q^4\in\sS_{\delta,P}$ for exactly those partitions $P$ that are coarser than (or equal to) $FP_{q^4}$ and finer than (or equal to) $P_{q^4}$.\label{fn:partition}}
\be
\sS_{\delta} = \bigcup_{P \in \mathcal{P}_N} \sS_{\delta,P}.
\ee

If $q^4\in \sS_{\delta,P}$ then we also write $t_\alpha$ for the joint time variable of all particles in $S_\alpha$, and $q_\alpha$ for the list of space coordinates of all particles belonging to $S_\alpha$; using that notation, we also write $q^4 = (t_1,q_1;\dots;t_L,q_L)$, so that $\sS_{\delta,P}$ can also be regarded as an open subset of $\RRR^{3N+L}$ (while we still sometimes find the notation $q^4=(x_1,\ldots,x_N)$ convenient). (The amibiguity as to whether the symbol $t_1$ refers to the family $\alpha=1$ or to the particle $j=1$ should not cause any practical difficulty.) We say that a function $\phi$ is \emph{smooth on $\sS_{\delta,P}$} if it is smooth as a function of the variables $t_1,q_1,\ldots,t_L,q_L$; we say that $\phi$ is \emph{smooth on $\sS_\delta$} if it is smooth on each $\sS_{\delta,P}$.

A potential $V$ is a pair potential if it can be written as
\begin{equation}
V = \sum_{\substack{i,j=1 \\ i \neq j}}^N W(\vect{x}_i - \vect{x}_j).
\end{equation}
We say that a pair potential has \emph{range $\delta$} if and only if, for all $\vx\in\RRR^3$,
\begin{equation}
W(\vx) = 0 \quad \text{for} \quad \|\vx\| \geq \delta.
\end{equation}

\begin{thm}\label{thm:delta_model}
Consider $N$ particles and a smooth pair potential with range $\delta>0$. Then the multi-time Dirac equation is consistent on the set \eqref{sSdeltadef} of $\delta$-spacelike configurations.

In more detail, let $W \in C^{\infty}(\RRR^{3}, \CCC^{4 \times 4})$ be a pair potential with range $\delta$. Then for all initial conditions $\phi_0\in C^{\infty}\left( \RRR^{3N}, (\CCC^4)^{\otimes N} \right)$ there is a unique wave function $\phi\in C^\infty\bigl( \sS_{\delta}, (\CCC^4)^{\otimes N} \bigr)$ with $\phi|_{t_1=\ldots=t_N=0} = \phi_0$ which satisfies on $\sS_{\delta,P}$ for every partition $P = \{ S_1,\dots,S_L \}$ the equations
\be\label{multi_time_Dirac_equation_N}
i\frac{\partial}{\partial t_{\alpha}} \phi(t_1, q_{1};\dots;t_L,q_L) =
\Biggl( \sum_{j \in S_{\alpha}} H_j^\free
 + \sum_{\substack{i,j \in S_{\alpha} \\ i \neq j}} W(\vect{x}_i - \vect{x}_j) \Biggr) \phi(t_1, q_{1};\dots;t_L,q_L)
\ee
for all $\alpha = 1,\dots,L$ and with $t_{\alpha} = x^0_i$ for all $i \in S_{\alpha}$; here, $H_j^\free$ is the free Dirac operator $-i \valpha_j \cdot \nabla_j + \beta_j m$.
\end{thm}

The proof is given in Section~\ref{sec:proof_delta}. Because of the shape of $\sS_\delta$, Theorems~\ref{thm:SchrN} and \ref{thm:SchrN_t} cannot be applied directly, and the check of consistency requires some care.

\section{Proofs of Results on Consistency Condition}\label{sec:proofs_consistency}

\begin{proof}[Proof of Theorem~\ref{thm:SchrN}]
For ease of notation let us set $N=2$ first. Let $H_1,H_2$ be given as self-adjoint operators on a Hilbert space $\Hilbert$ that do not depend on $t_1,t_2$. 
By the definition \eqref{strongdef} of ``strong solution,'' any strong solution $\phi$ of the multi-time equations \eqref{phiHj} must satisfy (cf.\ also Figure~\ref{fig:twopaths})
\be\label{eiH12phi}
e^{-iH_2t_2}e^{-iH_1t_1}\phi(0,0)=\phi(t_1,t_2) = e^{-iH_1t_1}e^{-iH_2t_2}\phi(0,0)\,,
\ee
and if a solution exists for every $\phi(0,0)\in\Hilbert$ then $e^{-iH_1t_1}$ must commute with $e^{-iH_2t_2}$, which requires that $[H_1,H_2]=0$ in the spectral sense. Conversely, $[H_1,H_2]=0$ in the spectral sense implies that $e^{-iH_1t_1}$ commutes with $e^{-iH_2t_2}$ for all $t_1$ and $t_2$, and then
\be
\phi(t_1,t_2):=e^{-iH_1t_1}e^{-iH_2t_2}\phi(0,0)
\ee
satisfies \eqref{strongdef}. For $N>2$ variables, one sees in the same way that it is necessary and sufficient for the existence of a strong solution that all $H_1,\ldots,H_N$ commute with each other, which proves Theorem~\ref{thm:SchrN}.
\end{proof}

Before proving Theorem~\ref{thm:SchrN_t}, we recall the Dyson expansion for the time evolution in the one-time case. We want to solve the Schr\"o\-din\-ger equation
\be\label{SchrHt}
i \frac{d \psi(t)}{d t} = H(t) \psi(t)
\ee
for arbitrary initial condition $\psi(s) \in \Hilbert$. It is known \cite[Theorem X.69]{reedsimon:1975} that for a strongly continuous function $t \mapsto H(t)$ into the bounded operators on $\Hilbert$ there is a unique operator $U(t,s)$ for every $s,t\in\RRR$ such that $U(s,s)=I$ and, for any $\psi(s)\in\Hilbert$, $t\mapsto\psi(t)=U(t,s)\psi(s)$ satisfies \eqref{SchrHt} with the derivative $d/dt$ understood as the limit in the Hilbert space topology of the difference quotient; $U(t,s)$ is given by a time-ordered exponential, the Dyson series
\begin{align}
U(t,s) &= \mathcal{T} e^{-i\int_s^t H(T) dT} \nonumber \\
&= I + \sum_{n=1}^{\infty} (-i)^n \int_{s}^{t} dT_1 \int_{s}^{T_1} dT_2 \cdots \int_{s}^{T_{n-1}} dT_n ~ H(T_1) H(T_2) \dots H(T_n),
\end{align}
which converges in operator norm (because it actually follows that $H(\cdot)$ is uniformly bounded on $[s,t]$), and satisfies, by the triangle inequality,
\be\label{Uineq1}
\|U(t,s)-I\| \leq \exp\Bigl( |t-s|\, \sup_{r\in[s,t]} \|H(r)\| \Bigr)-1 
\ee
and
\be
\label{Uineq2}
\|U(t,s)\| \leq \exp\Bigl( |t-s|\, \sup_{r\in[s,t]} \|H(r)\| \Bigr)\,.
\ee
We note further that always
\be\label{Texpcomposition}
U(t,s)U(s,r)=U(t,r)\,,
\ee
in particular
\be\label{Uinv}
U(s,t)U(t,s)=I\,,
\ee
so $U(t,s)$ is invertible. For self-adjoint $H(t)$ the operator $U(t,s)$ is unitary. It is perhaps useful to mention that the intuitive, or heuristic, meaning of the time-ordered exponential is a continuous product
\be\label{continuousproduct}
 \mathcal{T} e^{-i\int_s^t H(T) dT} \text{~``=''~} \prod_{T=s}^t \Bigl( I-i H(T)\, dT\Bigr)
\ee
in which the factors (which do not necessarily commute) are ordered so that $T$ increases from right to left.

We offer two proofs for Theorem~\ref{thm:SchrN_t} because we find both of them instructive.

\begin{proof}[First proof of Theorem~\ref{thm:SchrN_t}]
For ease of notation we formulate the proof for $N=2$, although the arguments apply to any $N$. For the multi-time equations \eqref{phiHj} for $\phi:\RRR^2\to\Hilbert$, relevant time evolution operators are given by time-ordered exponentials. As a consequence of \eqref{phiHj} for $j=1$, the time evolution in the first time variable is given by
\be\label{phiU1}
\phi(t_1,t_2) = U(t_1,s_1;t_2) \phi(s_1,t_2)
\ee
with
\begin{align}
&U(t_1,s_1;t_2) = \mathcal{T}_1 e^{-i\int_{s_1}^{t_1} H_1(T,t_2) dT} \nonumber \\
&= I + \sum_{n=1}^{\infty} (-i)^n \int_{s_1}^{t_1} dT_1 \int_{s_1}^{T_1} dT_2 \cdots \int_{s_1}^{T_{n-1}} dT_n ~ H_1(T_1,t_2) H_1(T_2,t_2) \dots H_1(T_n,t_2)\,.\label{Dyson1}
\end{align}
Since $(t_1,t_2)\mapsto H(t_1,t_2)$ is assumed to be smooth (with respect to the operator norm), it is in particular norm-continuous and thus uniformly bounded on $[s_1,t_1]\times \{t_2\}$ (or any other compact set in $\RRR^2$), so that the Dyson series converges (in operator norm). 
Similarly, if we keep the first time variable fixed, we find
\be\label{phiU2}
\phi(t_1,t_2) = U(t_1;t_2,s_2) \phi(t_1,s_2) = \mathcal{T}_2 e^{-i\int_{s_2}^{t_2} H_2(t_1,T) dT} \phi(t_1,s_2).
\ee
We have thus obtained the time evolution for vertical and horizontal line segments in $\RRR^2$. For any path $\gamma:[0,1]\to\RRR^N$ that is a concatenation of such segments, it follows further that
\be
\phi(\gamma(1))=U_\gamma\phi(\gamma(0))
\ee
with $U_\gamma$ the product of the time evolution operators for each segment, ordered from right to left as the segments are run through by $\gamma$.

As a consequence of \eqref{phiU1} and \eqref{phiU2}, every solution $\phi:\RRR^2\to\Hilbert$ of \eqref{phiHj} satisfies (in analogy to \eqref{eiH12phi}, cf.\ also Figure~\ref{fig:twopaths})
\be\label{U1U2phi}
U(t_1,s_1;t_2)U(s_1;t_2,s_2)\phi(s_1,s_2)=\phi(t_1,t_2)=U(t_1;t_2,s_2)U(t_1,s_1;s_2)\phi(s_1,s_2)
\ee
for all $t_1,t_2,s_1,s_2\in\RRR$.
Now fix $s_1,s_2\in\RRR$. Consistency, i.e., the existence of a solution for every $\phi(0,0)\in\Hilbert$, is equivalent to the existence of a solution for every $\phi(s_1,s_2)\in\Hilbert$. Indeed, assuming consistency, then for any given $\tilde\phi(s_1,s_2)\in\Hilbert$, there exists a solution with $\phi(0,0)=U(0,s_1;0)U(s_1;0,s_2)\tilde\phi(s_1,s_2)$, and by \eqref{U1U2phi} and \eqref{Uinv} it will have
\begin{align}
\phi(s_1,s_2)
&=U(s_1;s_2,0)U(s_1,0;0)\phi(0,0)\\
&=U(s_1;s_2,0)U(s_1,0;0)U(0,s_1;0)U(s_1;0,s_2)\tilde\phi(s_1,s_2)\\
&=\tilde\phi(s_1,s_2)\,.
\end{align}
Conversely, if there exists a solution for any choice of $\phi(s_1,s_2)$, then choose $\phi(s_1,s_2)=U(s_1;s_2,0)U(s_1,0;0)\tilde\phi(0,0)$, and by \eqref{U1U2phi} and \eqref{Uinv} the solution will have
\begin{align}
\phi(0,0)
&=U(0,s_1;0)U(s_1;0,s_2)\phi(s_1,s_2)\\
&=U(0,s_1;0)U(s_1;0,s_2)U(s_1;s_2,0)U(s_1,0;0)\tilde\phi(0,0)\\
&=\tilde\phi(0,0)\,,
\end{align}
so consistency follows.

As a consequence, consistency is equivalent to
\be\label{Utwopathsrect}
U(t_1,s_1;t_2)U(s_1;t_2,s_2)=U(t_1;t_2,s_2)U(t_1,s_1;s_2) \quad\forall t_1,t_2,s_1,s_2\in\RRR\,.
\ee
Put differently, the condition is that $U_\gamma=I$ for any path $\gamma$ around an axiparallel rectangle.

We now show that \eqref{Utwopathsrect} implies \eqref{consistency}. Rename $s_j\to t_j$, $t_j\to t_j+\Delta t$ with $\Delta t>0$. (Our arguments will also go through for negative $\Delta t$, but the presentation is simplified by assuming it is positive.) We write $O(\Delta t^n)$ for any operator $R$ (that may depend on $\Delta t$ and other things) such that there is a constant $0<M<\infty$ (independent of $\Delta t$ or other things) with $\|R\|\leq M\Delta t^n$ whenever $\Delta t$ is sufficiently small. We write $o(\Delta t^n)$ for any operator $R$ such that $\|R\|\leq f(\Delta t)$ for some function $f$ with $\lim_{\Delta t\to 0}\Delta t^{-n}f(\Delta t)=0$. Clearly, any $O(\Delta t^{n+1})$ is an $o(\Delta t^n)$. For any $(s_1,s_2)\in [t_1,t_1+\Delta t]\times [t_2,t_2+\Delta t]$, we can write
\be\label{Hjexpand}
H_j(s_1,s_2) = H_j(t_1,t_2) + \sum_{k=1}^2(s_k-t_k)\frac{\partial H_j}{\partial t_k}(t_1,t_2) + o(\Delta t)
\ee
because $H_j$ was assumed to be a differentiable function of $(t_1,t_2)$.
We can write the Dyson series \eqref{Dyson1} in the form
\begin{align}
U(t_1+\Delta t,t_1;t_2) &= I -i \int_{t_1}^{t_1+\Delta t} dT_1~ H_1(T_1,t_2) \nonumber\\
&\quad - \int_{t_1}^{t_1+\Delta t} dT_1 \int_{t_1}^{T_1}dT_2~H_1(T_1,t_2)H_1(T_2,t_2) + O(\Delta t^3)\,.\label{Dyson2}
\end{align}
Indeed, the remainder $R$ comprising all terms of order $n\geq 3$ in the Dyson series satisfies, in analogy to \eqref{Uineq1}, $\|R\| \leq f\bigl( M\Delta t  \bigr)$ with $f(x)=e^x-1-x-\frac{1}{2}x^2$, $\Delta t\leq 1$, and
\be
M=\sup \Bigl\{ \|H_j(s_1,s_2)\|: j=1,2, \: s_1\in[t_1,t_1+1], s_2\in[t_2,t_2+1] \Bigr\} < \infty.
\ee
Since $f(x)\leq x^3$ for sufficiently small positive $x$, $R$ is an $O(\Delta t^3)$ and an $o(\Delta t^2)$.

Plugging \eqref{Hjexpand} into \eqref{Dyson2}, we obtain that
\begin{align}
U(t_1+\Delta t,t_1;t_2)
&=I -i H_1(t_1,t_2)\Delta t - \frac{i}{2} \frac{\partial H_1}{\partial t_1}(t_1,t_2)\,\Delta t^2 + o(\Delta t^2)\nonumber\\
&\quad - \frac{1}{2} H_1(t_1,t_2)^2 \, \Delta t^2 + o(\Delta t^2)\,.\label{UDeltat1}
\end{align}
Likewise, abbreviating $H_j(t_1,t_2)$ by $H_j$ and $\frac{\partial H_j}{\partial t_k}(t_1,t_2)$ by $\frac{\partial H_j}{\partial t_k}$,
\begin{align}
U(t_1+\Delta t,t_1;t_2+\Delta t)
&=I -i H_1\Delta t - i \frac{\partial H_1}{\partial t_2} \, \Delta t^2- \frac{i}{2} \frac{\partial H_1}{\partial t_1}\,\Delta t^2 - \frac{1}{2} H_1^2 \, \Delta t^2 + o(\Delta t^2)\label{UDeltat2}\\
U(t_1;t_2+\Delta t,t_2)
&=I -i H_2\Delta t - \frac{i}{2} \frac{\partial H_2}{\partial t_2}\,\Delta t^2 - \frac{1}{2} H_2^2 \, \Delta t^2 + o(\Delta t^2)\label{UDeltat3}\\
U(t_1+\Delta t;t_2+\Delta t,t_2)
&=I -i H_2\Delta t - i \frac{\partial H_2}{\partial t_1} \, \Delta t^2- \frac{i}{2} \frac{\partial H_2}{\partial t_2}\,\Delta t^2 - \frac{1}{2} H_2^2 \, \Delta t^2 + o(\Delta t^2)\,.\label{UDeltat4}
\end{align}
A simple calculation shows that the difference between the left and the right-hand side of \eqref{Utwopathsrect} can be expressed as
\begin{align}
0&=U(t_1+\Delta t,t_1;t_2+\Delta t)U(t_1;t_2+\Delta t,t_2)-U(t_1+\Delta t;t_2+\Delta t,t_2)U(t_1+\Delta t,t_1;t_2)\nonumber\\
&= \biggl( -[H_1,H_2] -i\frac{\partial H_1}{\partial t_2} + i \frac{\partial H_2}{\partial t_1} \biggr) \Delta t^2 + o(\Delta t^2)\,.\label{Udiff}
\end{align}
Therefore, the bracket in front of $\Delta t^2$ must vanish. (After all, the equation is true for every sufficiently small $\Delta t$; if the bracket did not vanish, then for small enough $\Delta t$ the $o(\Delta t^2)$ would be too small to cancel it.) This proves \eqref{consistency2} or, equivalently, \eqref{consistency}. 

Conversely, suppose that the bracket vanishes at all points in an axiparallel rectangle, which we can take to be $[0,t_1]\times [0,t_2]$. Subdivide the rectangle into small squares of side length $\Delta t$, and for each square, consider the paths shown in Figure~\ref{fig:aroundsquare}; call them $\gamma$ and $\gamma'$, and let the square be $[s_1,s_1+\Delta t]\times [s_2,s_2+\Delta t]$.

\begin{figure}[ht]
\centering
\includegraphics[width=200pt,keepaspectratio]{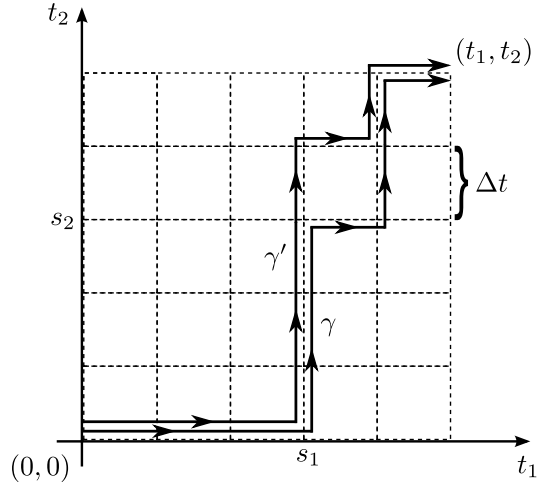}
\centering
\caption{\small{Two paths from $(0,0)$ to $(t_1,t_2)$ that pass along different sides of a small square. The paths actually lie \emph{on} dashed lines and are drawn next to them only for better visibility.}}
\label{fig:aroundsquare}
\end{figure}

We know from the equality of the two right-hand sides in \eqref{Udiff} that, when the bracket vanishes, the southeast and the westnorth way around a square differ by $o(\Delta t^2)$; thus, $U_{\gamma'}-U_{\gamma}$ is of the form $U(t_1,s_1+\Delta t;t_2)U(s_1+\Delta t;t_2,s_2+\Delta t)o(\Delta t^2)U(s_1;s_2,0)U(s_1,0;0)$; by \eqref{Uineq2}, this operator has norm no greater than $e^{(t_1+t_2)M}\|o(\Delta t^2)\|$ with
\be
M=\sup\Bigl\{\|H_j(s_1,s_2)\|:j=1,2,s_1\in[0,t_1],s_2\in[0,t_2]\Bigr\}<\infty\,,
\ee
and thus is itself an $o(\Delta t^2)$. Now sum $U_{\gamma'}-U_{\gamma}$ over all squares; due to cancellations (sum in the order depicted in the right diagram of Figure~\ref{fig:Sigma}), the result $R$ is the difference of the $U$ of the two paths shown in Figure~\ref{fig:twopaths},
\be
R=U(t_1,0;t_2)U(0;t_2,0)-U(t_1;t_2,0)U(t_1,0;0).
\ee
Since the number of small squares is proportional to $\Delta t^{-2}$, $R$ is an $o(1)$; letting $\Delta t\to 0$, we see that $R=0$. Thus, we have shown \eqref{Utwopathsrect} (our simplifying assumption $s_1=s_2=0$ can be dropped), and thus consistency.
\end{proof}

\begin{proof}[Second proof of Theorem~\ref{thm:SchrN_t}]
The theme of this proof is the path independence described in Section~\ref{sec:pathindependence}. We begin by showing that also for other paths $\gamma$, not just horizontal and vertical line segments, there exist time evolution operators given by suitable ordered exponentials. 

So  consider an arbitrary smooth path $\gamma:[0,1] \to \RRR^N, t \mapsto (\gamma_1(t),\ldots,\gamma_N(t))$. Our first claim is that, for any solution $\phi$ of \eqref{phiHj}, the function $t\mapsto \phi(\gamma(t))$ is differentiable, and its derivative is, as expected from the chain rule,
\begin{align}
\frac{d}{dt} \phi(\gamma(t)) 
&= \sum_{j=1}^N \frac{\partial \phi}{\partial t_j} \dot\gamma_j\label{chain}\\
&=\Bigl(-i\sum_{j=1}^N H_j\bigl(\gamma(t)\bigr)\dot\gamma_j\Bigr) \phi(\gamma(t))\label{chain2}\,.
\end{align}
with $\dot{\gamma}(t) = d \gamma(t)/dt$. Note that the chain rule applies only if $\phi$ is a differentiable function of $\vec{t}=(t_1,\ldots,t_N)$, and we have not yet shown this to be the case.\footnote{A well-known theorem asserts that a real-valued function of $\vec{t}$ is differentiable as soon as it possesses continuous partial derivatives. However, while our $\phi$ possesses partial derivatives according to \eqref{phiHj}, which can be shown with a little effort to be continuous, this theorem yields only that, for every fixed vector $\chi\in\Hilbert$, $\vec{t}\mapsto\scp{\chi}{\phi(\vec{t})}$ is continuous, and not that $\vec{t}\mapsto\phi(\vec{t})$ is continuous in the norm topology, as needed.} So we now prove differentiability or, what amounts to the same, the chain rule \eqref{chain}.

Let $\vec{t}=\gamma(t)$, $\Delta t_j=\gamma_j(t+\Delta t)-\gamma_j(t)$, $\Delta t>0$, and $c=\sup\bigl\{|\dot\gamma_j(s)|:j\in\{1\ldots N\},s\in[0,1]\bigr\}<\infty$; note that $|\Delta t_j|\leq c\Delta t$. We consider the piecewise axiparallel path from $\gamma(t)$ to $\gamma(t+\Delta t)$ that first changes the first variable, then the second, etc.; by \eqref{phiU1}, we have that
\be\label{phiUNU1phi}
\phi(\gamma(t+\Delta t)) = U_N\cdots U_1 \phi(\gamma(t))
\ee
with $U_j=U\bigl(t_N;\ldots;t_{j+1};t_j+\Delta t_j,t_j;t_{j-1}+\Delta t_{j-1};\ldots;t_1+\Delta t_1\bigr)$ given by a Dyson series as in \eqref{Dyson1}. Using expansions as in \eqref{UDeltat1}--\eqref{UDeltat4}, based on truncating the Dyson series as in \eqref{Dyson2}, but this time expanding only to first order, we obtain that
\be\label{chainDeltat}
U_N\cdots U_1=I-i\sum_j H_j \Delta t_j  + o(\Delta t)\,,
\ee
using that every $o(|\Delta t_j|)$ is an $o(\Delta t)$ because $|\Delta t_j|\leq c\Delta t$.
Applying \eqref{chainDeltat} to $\phi(\gamma(t))$, we obtain from \eqref{phiUNU1phi} that
\be
\lim_{\Delta t\to 0}\frac{\phi(\gamma(t+\Delta t))-\phi(\gamma(t))}{\Delta t} = -i\sum_j H_j(\gamma(t)) \,\dot\gamma_j(t) \, \phi(\gamma(t))\,,
\ee
or \eqref{chain2}.

Now that we have \eqref{chain2}, this can be expressed by saying that $\phi\circ\gamma$ satisfies a Schr\"o\-din\-ger equation with $t$-dependent Hamiltonian
\be
H(t)=\sum_{j=1}^N H_j\bigl(\gamma(t)\bigr)\dot\gamma_j(t)\,.
\ee
Therefore, the solution is given by the appropriate $t$-ordered exponential,
\be\label{phiUgamma}
\phi(\gamma(1))=U_\gamma \phi(\gamma(0))
\ee
with
\begin{align}\label{U_gamma}
U_{\gamma} &= \mathcal{T} e^{-i\int_{\gamma} \sum_j H_j dt_j} \nonumber \\
&= 1 + \sum_{n=1}^{\infty} (-i)^n \int_0^1 dT_1 \int_0^{T_1} dT_2 \cdots \int_0^{T_{n-1}} dT_n ~\times\nonumber\\
&\quad \quad \times ~ \biggl(\sum_j H_j(\gamma(T_1))\dot{\gamma}_j(T_1)\biggr) \cdots 
\biggl(\sum_j H_j(\gamma(T_n))\dot{\gamma}_j(T_n) \biggr)\,.
\end{align}
The conclusion \eqref{phiUgamma} is also true for paths that are \emph{piecewise} smooth because if $\gamma,\gamma'$ are smooth paths with $\gamma'(0)=\gamma(1)$ then \eqref{phiUgamma} implies that $\phi(\gamma'(1))=U_{\gamma'}U_\gamma\phi(\gamma(0))$, while the property \eqref{Texpcomposition} implies that $U_{\gamma'}U_\gamma$ equals the Dyson series associated with the concatenation of $\gamma$ and $\gamma'$.

As a corollary of the results so far, we obtain a certain kind of path independence: that for any two paths $\gamma,\gamma'$ connecting two points $\vec{t},\vec{t}^{\,\prime}\in\RRR^N$ and any solution $\phi$ of \eqref{phiHj}, $U_\gamma \phi\bigl(\vec{t}\,\bigr)=U_{\gamma'}\phi\bigl(\vec{t}\,\bigr)$. As a corollary of \emph{that}, we obtain that if the multi-time equations \eqref{phiHj} are consistent (i.e., possess a solution for every $\phi(0,\ldots,0)$), then the path independence discussed in Section~\ref{sec:pathindependence} holds, i.e., $U_\gamma=U_{\gamma'}$ for any two paths from $(0,\ldots,0)$ to $\vec{t}$ (and thus also for any two paths between two given points). (Conversely, if path independence holds, we know already that the system \eqref{phiHj} is consistent.) It remains to show that path independence is equivalent to \eqref{consistency}.

\medskip

We now assume the point of view described in Section~\ref{sec:pathindependence}, regarding $\phi$ as a cross section of a vector bundle with base space $\RRR^N$, fibers $\Hilbert$, connection given by the $H_j$, and parallel transport operator $U_\gamma$. 
We now show that the parallel transport is path-independent if and only if the connection is \emph{flat}, i.e., its curvature vanishes. The curvature is a two-form $F$ with values in the Lie algebra of the gauge group (here, with values in the bounded operators on $\Hilbert)$; we have given the formula for $F$ in \eqref{Fjkdef}.

The non-Abelian Stokes theorem \cite{arafeva:1980,bralic:1980} for parallel transport in a vector bundle expresses the holonomy (i.e., parallel transport along a closed path $\gamma$) in terms of the curvature of the connection integrated over a 2-dimensional surface $\Sigma$ whose boundary is $\gamma$. It asserts that, for any 2-surface $\Sigma$ parameterized by a $C^1$ function $f:[0,1]^2\to\RRR^N$,
\be\label{Stokes}
\mathcal{T} \exp\Biggl(-i \int\limits_{\partial \Sigma} \sum_j H_j\, dt_j\Biggr) 
= \mathcal{P}_{\!f} \exp \Biggl(i\int\limits_\Sigma \sum_{i,j} \mathcal{F}_{ij}\,dt_i\wedge dt_j\Biggr)\,. 
\ee

Before we apply this formula to the case at hand, we elucidate it. The left-hand side, a path-ordered integral over the connection coefficients, is, in fact, equal to the holonomy along the boundary curve of $\Sigma$. For a heuristic understanding of this fact, think of parallel transport from $\gamma(t)$ to $\gamma(t+dt)$ (relative to the reference connection that identifies the fibers with $\Hilbert$) as the application of the operator $I-i\sum_j H_j \, \dot\gamma_j \, dt = I-i\sum_j H_j \, dt_j$. Thus, parallel transport along the whole path $\gamma$ is the application of the ``continuous product'' 
\be
\prod_{t=0}^1 \Bigl( I-iH(t) dt \Bigr) \quad \text{with }H(t)=\sum_j H_j(\gamma(t))\dot{\gamma}_j(t)\,,
\ee
with the (non-commuting) factors ordered so that $t$ increases from right to left. By \eqref{continuousproduct}, this is the same as the path-ordered exponential integral of the connection coefficients, i.e., the left-hand side of \eqref{Stokes}. 

The right-hand side of \eqref{Stokes} is a suitably ordered exponential surface integral of an operator-valued 2-form $\mathcal{F}=\sum_{ij}\mathcal{F}_{ij}\, dt_i\wedge dt_j$ obtained from the curvature 2-form $F$ by suitable parallel transport to the reference point $\gamma(0)$ at which the boundary curve begins. The right-hand side of \eqref{Stokes} can be regarded heuristically as the two-parameter continuous product
\be
\prod_{s=0}^1 \prod_{t=0}^1\Bigl( I+i\sum_{i,j}\mathcal{F}_{ij} \frac{\partial f_i}{\partial s} \frac{\partial f_j}{\partial t} \, ds\,dt  \Bigr)
\ee
with the terms ordered from right to left according to the right diagram of Figure~\ref{fig:Sigma}, and with the operator $\mathcal{F}_{ij}(s,t)$ obtained from $F_{ij}(f(s,t))$ by parallel transport along the image under $f$ of the path $(s,t)\to(s,0)\to(0,0)$ in the $st$-plane.

\begin{figure}[ht]
\centering
\includegraphics[width=400pt,keepaspectratio]{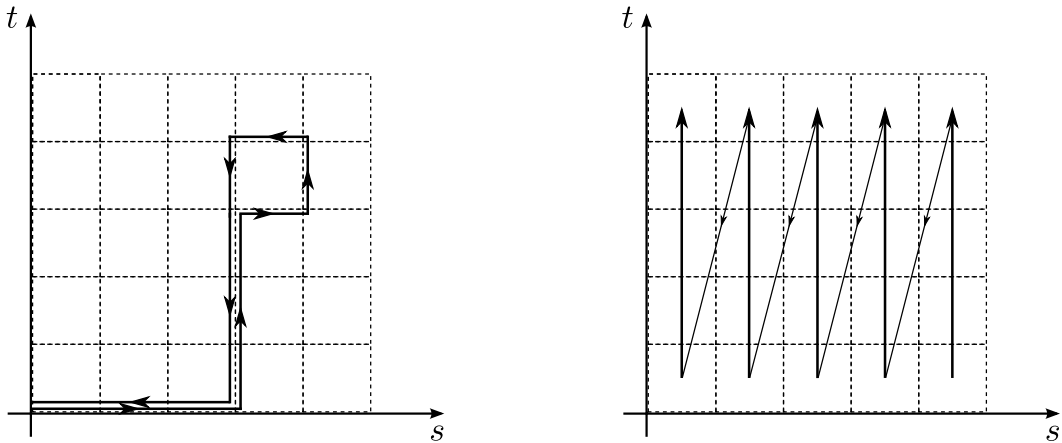}
\centering
\caption{\small{Heuristic derivation of the non-Abelian Stokes theorem, based on subdividing the 2-surface $\Sigma$ (parameterized by $s,t$) into tiny rectangles. Left: For each rectangle, consider this path surrounding it. Right: Ordering of rectangles that corresponds to order of contributions in the surface integral. }}
\label{fig:Sigma}
\end{figure}

To understand heuristically why \eqref{Stokes} is true, one may first convince oneself that for every \emph{infinitesimal} rectangle $R$, the holonomy around $\partial R$ (counterclockwise in the $st$-plane) equals
\be\label{infinitesimal_holonomy}
I + i\int_R \sum_{ij}F_{ij}\, dt_i\wedge dt_j 
= I + i\sum_{ij} F_{ij} \frac{\partial f_i}{\partial s} \frac{\partial f_j}{\partial t} \, ds\,dt + o(ds\,dt)\,.
\ee
For a finite-sized surface $\Sigma$, subdivide it into tiny rectangles $R$, let $\gamma_R$ be the image under $f$ of the path shown in the left diagram of Figure~\ref{fig:Sigma}, let $U_{\gamma_R}$ be the holonomy of that path as given in \eqref{infinitesimal_holonomy}, and let $\gamma$ be the concatenation of the $\gamma_R$ for all $R$ in the order shown in the right diagram of Figure~\ref{fig:Sigma}. Note that in $U_\gamma$ all inner segments cancel out, so that $U_{\gamma}=U_{\gamma'}$ with $\gamma'$ the path that surrounds the big square $[0,1]^2$ once counterclockwise. Now \eqref{Stokes} follows.\footnote{It may seem surprising that, in \eqref{Stokes}, $F_{ij}$ must be parallel-transported to the reference point $f(0,0)$ but $H_j$ need not. That is because $H_j$ actually expresses the difference between two gauge connections, the one [$\nabla$ in \eqref{covder}] representing the time evolution and the trivial one [$\partial$ in \eqref{covder}] corresponding to the identification of all fibers with $\Hilbert$; as a consequence, the trivial connection must be used to transport $H_j$ to $f(0,0)$. On the other hand, the inner segments in $U_\gamma=\prod_R U_{\gamma_R}$ cancel only if the same connection $\nabla$ is used along all pieces of $\gamma$, so the non-trivial connection $\nabla$ must be used to transport $F_{ij}$ to $f(0,0)$. By the way, the reference connection need not be mentioned in \eqref{Stokes} if we understand the left-hand side as the holonomy of $\nabla$ over $\partial \Sigma$.}

For the sake of completeness, we also express the definition of the right-hand side of \eqref{Stokes} in a formula (following the notation of \cite{arafeva:1980}):\footnote{It may seem surprising that the $T_m$ ($m=1,\ldots,n$) are not ordered decreasingly. Take note that the factors involving $(S_m,T_m)$ must be ordered according to the right diagram in Figure~\ref{fig:Sigma}, i.e., terms with bigger $S_m$ must be further to the right; only when $S_m=S_{m+1}$ then terms must be ordered by $T$, but this demand can be ignored because $S_m=S_{m+1}$ occurs only on a set of measure zero.}
\begin{align}
&\mathcal{P}_{\!f} \exp \Biggl(i\int\limits_\Sigma \sum_{i,j} \mathcal{F}_{ij}\,dt_i\wedge dt_j\Biggr) =\nonumber\\
&\quad I + \sum_{n=1}^{\infty}i^n \int_0^1 dS_1 \int_0^1 dT_1 \int_0^{S_1} dS_2 \int_0^1 dT_2 \cdots \int_0^{S_{n-1}} dS_n \int_0^1 dT_n ~\times\nonumber\\[4mm]
&\quad \times~ (\mathcal{F}f'\dot{f})(S_n,T_n) \dots (\mathcal{F}f'\dot{f})(S_1,T_1),
\end{align}
with $(\mathcal{F}f'\dot{f})(s,t) = \sum_{i,j}\mathcal{F}_{ij}(s,t)\, f_i'(s,t)\, \dot{f}_j(s,t)$, where $f'_i(s,t)=\partial f_i(s,t)/\partial s$ and $\dot{f}_j(s,t)=\partial f_j(s,t)/\partial t$. Moreover,
\begin{align}
\mathcal{F}_{ij}(s,t) &= h^{-1}(s,0) \, g^{-1}(s,t) \, F_{ij}(f(s,t)) \, g(s,t) \, h(s,0),\\[3mm]
h(s,t) &= \mathcal{T} \exp\Biggl(-i \int_0^s \sum_i H_i(f(s',t))\, f'_i(s',t)\, ds'\Biggr),\\
g(s,t) &= \mathcal{T} \exp\Biggl(-i \int_0^t \sum_i H_i(f(s,t'))\, \dot{f}_i(s,t') \,dt' \Biggr).
\end{align}
The proof of the non-Abelian Stokes formula \eqref{Stokes} in \cite{arafeva:1980} is formulated for finite-dimensional fiber spaces but remains valid if $H_1,\dots,H_N$ are, as in our case, bounded operators on $\Hilbert$.

We now apply \eqref{Stokes} to our connection with coefficients $A_j=-H_j$. Since $\RRR^N$ is simply connected, every closed path $\gamma$ actually is the boundary curve of a 2-surface $\Sigma$. Now we see immediately from \eqref{Stokes} that vanishing $F_{ij}$ for all $i,j$ implies that for every closed path $\gamma$ the left-hand side of \eqref{Stokes} is $I$, which is equivalent to path independence. 
Conversely, if $F_{ij} \neq 0$ for some $i,j$ at some point $(s_1,\dots,s_N)$, then, due to the smoothness of $H_1,\dots,H_N$, there is a whole neighborhood of $(s_1,\dots,s_N)$ where $F_{ij} \neq 0$. Then there is also a closed path $\gamma=\partial \Sigma$ such that \eqref{Stokes} does not vanish, and so path independence fails. We have already seen around \eqref{Utwopathsrect} that path independence (already merely for boundaries of axiparallel rectangles) is also sufficient for consistency.
\end{proof}

\section{Proofs of Inconsistency Theorems}\label{sec:proofs_inconsistentV}

\begin{proof}[Proof of Theorems \ref{thm:inconsistentV1} and \ref{thm:inconsistentV2}]
We directly prove Theorem~\ref{thm:inconsistentV2}, since Theorem~\ref{thm:inconsistentV1} follows for a special choice of $A_{i,1},A_{i,2},A_{i,3},B_i$. Since the $H_i^\free$ commute with each other, and the $V_i$ do,
the consistency condition \eqref{consistency} is
\begin{align}
0&=-i\frac{\partial V_j}{\partial t_i} + i\frac{\partial V_i}{\partial t_j} +[H_i,H_j] \\
&= -i\Bigl(\frac{\partial V_j}{\partial t_i} -\frac{\partial V_i}{\partial t_j}\Bigr) + [H_i^\free,V_j] - [H_j^\free,V_i] \\
&= -i\Bigl(\frac{\partial V_j}{\partial t_i} -\frac{\partial V_i}{\partial t_j}\Bigr)
-i\sum_{a=1}^3 \Bigl( A_{i,a}\frac{\partial V_j}{\partial x_{i,a}} - A_{j,a} \frac{\partial V_i
}{\partial x_{j,a}} \Bigr).
\end{align}
Since $A_{i,a}$ and $A_{j,a}$ act on different indices (viz., $s_i$ and $s_j$), the seven matrices
\be\label{sevenmatrices}
I,A_{i,1},A_{i,2},A_{i,3},A_{j,1},A_{j,2},A_{j,3}
\ee
are linearly independent when regarded as $k\times k$ matrices ($k=\prod_{\ell=1}^N k_\ell$) acting on all spin indices $s_1,\ldots,s_N$. Thus,
\be\label{relation1}
\frac{\partial V_j}{\partial t_i} -\frac{\partial V_i}{\partial t_j} = 0
\ee
and
\be\label{relation2}
\frac{\partial V_j}{\partial x_{i,a}}=0\,,
\ee
so $V_j$ does not depend on $\vx_i$; since the same argument applies to any $i\neq j$, we have that $V_j=V_j(\vx_j,t_1,\ldots,t_N)$. Set
\be\label{tilde_relation}
\tilde{V}_j(\vx_j,t_j) := V_j(\vx_j,t_1=0,\ldots,t_{j-1}=0,t_j,t_{j+1}=0,\ldots,t_N=0)
\ee
for every $j$. By \eqref{relation1}, for $i\neq j$ and any $a=1,2,3$,
\begin{align}
\frac{\partial}{\partial t_i} \frac{\partial}{\partial x_{j,a}}\bigl(V_j-\tilde{V}_j\bigr)
&= \frac{\partial}{\partial x_{j,a}}\frac{\partial}{\partial t_i} \bigl(V_j-\tilde{V}_j\bigr)\\
&= \frac{\partial}{\partial x_{j,a}}\frac{\partial V_j}{\partial t_i} \\
&= \frac{\partial}{\partial x_{j,a}}\frac{\partial V_i}{\partial t_j}\\
&= \frac{\partial}{\partial t_j} \frac{\partial V_i}{\partial x_{j,a}}\\
&=0\,.
\end{align}
Thus, $\frac{\partial}{\partial x_{j,a}}(V_j-\tilde{V}_j)$ does not depend on $t_i$ for any $i\neq j$, and since it vanishes when $t_i=0$ for all $i\neq j$, we have that $\frac{\partial}{\partial x_{j,a}}(V_j-\tilde{V}_j)=0$ for all $\vx$'s and $t$'s. Thus,
\be
V_j(\vx_j,t_1,\ldots,t_N) = \tilde{V}_j(\vx_j,t_j) + W_j(t_1,\ldots,t_N)
\ee
for some real-valued, smooth function $W_j$. Since for $i\neq j$,
\be
\frac{\partial V_j}{\partial t_i} = \frac{\partial \tilde{V}_j}{\partial t_i} + \frac{\partial W_j}{\partial t_i}=\frac{\partial W_j}{\partial t_i}\,,
\ee
we obtain from \eqref{relation1} that
\be\label{relation3}
\frac{\partial W_j}{\partial t_i} - \frac{\partial W_i}{\partial t_j} =0\,,
\ee
so $W_j= \partial \theta/\partial t_j$ for some real-valued, smooth function $\theta(t_1,\ldots,t_N)$. From this, \eqref{inconVfree} follows.
\end{proof}

In order to prove Theorem~\ref{thm:inconsistentV3} we need the following auxiliary lemma.

\begin{lem}\label{lem:sl_xjt_conn}
For any $j\in\{1,\ldots,N\}$, $\vx\in\RRR^3$, and $t_1,\ldots,t_N\in\RRR$, the set
\be
\sS_{j,\vx,t_1,\ldots,t_N} = \Bigl\{ (t_1,\vx_1,\ldots,t_N,\vx_N) \in \sS_{\neq}: \vx_j=\vx \Bigr\}
\ee
is path-connected.
\end{lem}

\begin{proof}[Proof of Lemma~\ref{lem:sl_xjt_conn}]
First, note that $\sS_{j,\vx,t_1,\ldots,t_N} \subset \RRR^{4N}$ can be identified in an obvious way with
\begin{multline}
\tilde{\sS}=\tilde{\sS}_{j,\vx,t_1,\ldots,t_N}=\Bigl\{ (\vx_1, \ldots, \vx_{j-1}, \vx_{j+1}, \ldots, \vx_N) \in \RRR^{3N-3}: \\
\|\vx_i-\vx_k\| > |t_i-t_k| ~\forall i,k \in 1,\ldots,N, i\neq k \Bigr\}\,.
\end{multline}
Now choose two arbitrary elements
\be
x = (\vx_1, \ldots, \vx_{j-1}, \vx_{j+1}, \ldots, \vx_N)\text{ and }
y = (\vy_1, \ldots, \vy_{j-1}, \vy_{j+1}, \ldots, \vy_N)
\ee
in $\tilde{\sS}$. Each element is given by $N-1$ points in $\RRR^3$ that each have a fixed minimum distance to each other and to a fixed point $\vx_j$. One can connect $x$ and $y$ by the following path. First, choose a large number $\lambda$ and move the $\vx_1, \ldots, \vx_{j-1}, \vx_{j+1}, \ldots, \vx_N$ away from $\vx_j$ so that their distances grow continuously by the factor $\lambda$. That is, for the path parameter $0 \leq s \leq \lambda-1$, set $\vx_i(s) = \vx_j + (1+s)(\vx_i-\vx_j)$ for $i \neq j$. That path stays in $\tilde{\sS}$. Call the final configuration $x'= (\vx'_1, \ldots, \vx'_{j-1}, \vx'_{j+1}, \ldots, \vx'_N)$. One does the same procedure with the element $y$ to obtain $y'$. To find a path in $\tilde{\sS}$ from $x'$ to $y'$, it is useful to look at $x'$ and $y'$ zoomed out by the factor $\lambda$; the net effect of going from $x$ to $x'$ and then zooming out is that the minimum distances have shrunk by a factor $\lambda$. It is clear that the space $^N\RRR^3$ of $N$-element subsets of $\RRR^3$ is connected, and that there is a path in $^N\RRR^3$ from $\lambda^{-1}\{ \vx'_1,\ldots,\vx'_{j-1},\vx_j,\vx'_{j+1},\ldots,\vx'_N \}$ to $\lambda^{-1}\{ \vy'_1,\ldots,\vy'_{j-1},\vx_j,\vy'_{j+1},\ldots,\vy'_N \}$ that keeps $\lambda^{-1}\vx_j$ fixed. Choose such a path. It has a nonzero minimal distance $\delta$ that any two points ever reach, so if $\lambda$ is so large that $\max_{i,k=1,\ldots,N} \lambda^{-1}|t_i-t_k|<\delta$ then the corresponding path from $x'$ to $y'$ stays in $\tilde{\sS}$. In this way we have constructed a path that connects two arbitrary elements $x$ and $y$ from $\tilde{\sS}$.
\end{proof}

\begin{proof}[Proof of Theorem~\ref{thm:inconsistentV3}]
As shown in the proof of Theorem~\ref{thm:inconsistentV2}, the consistency condition is equivalent to the conjunction of
\be\label{relation1_new}
\frac{\partial V_j}{\partial t_i} -\frac{\partial V_i}{\partial t_j} = 0
\ee
and
\be\label{relation2_new}
\frac{\partial V_j}{\partial x_{i,a}}=0
\ee
on $\sS_{\neq}$ for all $i, j$ ($i \neq j$). For fixed $j$, equation \eqref{relation2_new} holds for all $i \neq j$ and $a=1,2,3$, i.e., $V_j$ is a function with vanishing gradient on the set $\sS_{j,\vx_j,t_1,\ldots,t_N}$, which is pathwise connected according to Lemma~\ref{lem:sl_xjt_conn}. Therefore we can conclude from \eqref{relation2_new} that $V_j$ is constant on $\sS_{j,\vx_j,t_1,\ldots,t_N}$; put differently, $V_j$ does not depend on $\vx_1, \ldots, \vx_{j-1}, \vx_{j+1}, \ldots, \vx_N$. Since for every $\vx_j,t_1,\ldots,t_N$, the set $\sS_{j,\vx_j,t_1,\ldots,t_N}$ is non-empty, $V_j(\vx_j,t_1,\ldots,t_N)$ is indeed a function on all of $\RRR^{3+N}$. Thus all steps in the proof of Theorem~\ref{thm:inconsistentV2} from Equation~\eqref{tilde_relation} onwards remain valid, and we obtain that \eqref{inconVfree} holds on $\sS_{\neq}$. 
\end{proof}

In order to prove Theorem~\ref{thm:inconsistentV4} we need two more auxiliary lemmas.

\begin{lem}\label{lem:sl_conn}
The set $\sS_{\neq} \subset \RRR^{4N}$ as in \eqref{sSneqdef} of ordered spacelike configurations without collisions in $3+1$ dimensions is simply connected.
\end{lem}

(It is interesting to note that this does not hold in $2+1$ or $1+1$ dimensions.)

\begin{proof}[Proof of Lemma~\ref{lem:sl_conn}]
We first show that $\sS_{\neq}$ is path-connected. Choose two arbitrary points $x^{4N}=(x_1,\ldots,x_N) \in \sS_{\neq}$ and $y^{4N}=(y_1,\ldots,y_N) \in \sS_{\neq}$. We now construct a path from $x^{4N}$ to $y^{4N}$ within $\sS_{\neq}$. First, move the time variables $x_1^0,\ldots,x_N^0$ to the common time $x_1^0,\ldots,x_N^0=0$ along the path
\be
\gamma:[0,1]\to \RRR^{4N}, \gamma(s)=\Bigl( \bigl( (1-s)x_1^0,\vx_1 \bigr),\ldots, \bigl( (1-s)x_N^0,\vx_N \bigr) \Bigr).
\ee
This path stays in $\sS_{\neq}$ since $\|\vx_i-\vy_j\| > |x_i^0-x_j^0| \geq (1-s)|x_i^0-x_j^0|$ for $s \in [0,1]$ and all $i \neq j$. We have thus obtained a path from $x^{4N}$ to $x^{4N}_0=((0,\vx),\ldots,(0,\vx_N))$ that stays in $\sS_{\neq}$. Do the same with $y^{4N}$ to obtain $y^{4N}_0$. Now one can easily move $\vx_1$ to $\vy_1$ by any path that avoids the points $\vx_2,\ldots,\vx_N$. For example, one could choose a straight line $(1-s')\vx_1+s'\vy_1$, $s' \in [0,1]$, and in case some of the $\vx_2,\ldots,\vx_N$ are on this line, one modifies the path by going around these points in a semicircle. After that, move $\vx_2$ to $\vy_2$ along some path that avoids $\vy_1,\vx_3,\ldots,\vx_N$ and repeat this procedure for $\vx_3,\ldots,\vx_N$. This shows that $\sS_{\neq}$ is path connected.

Now we show that any closed path $(x_1(t),\ldots,x_N(t))$ (with $(x_i(0)=x_i(1)$) in $\sS_{\neq}$ can be deformed to a point. As we did for single points in $\sS_{\neq}$ above, we now move the whole path to the $x_1^0,\ldots,x_N^0=0$ plane. Note that the path thus obtained in $\RRR^{3N}$ corresponds to $N$ paths in $\RRR^3$ with the property that for any fixed $t$ the diagonals $\vx_i(t)=\vx_j(t)$ are avoided for all $i \neq j$. Now one can indeed deform each of the paths in $\RRR^3$ to a point by an arbitrary deformation. Should it occur that for some path parameter $t$, some $i \neq j$ and some deformation parameter $s$, $\vx^{(s)}_i(t)=\vx^{(s)}_j(t)$, then just change the deformation such that the $i$-th path is deformed faster, so the diagonals are avoided during the deformation. This procedure shows that any path can be contracted within $\sS_{\neq}$.
\end{proof}

\begin{lem}\label{lem:sl_xj_conn}
For any $j\in\{1,\ldots,N\}$ and $x\in\RRR^4$, the set $\sS_{j,x} := \{ (x_1,\ldots,x_N) \in \sS_{\neq}: x_j = x \}$ is path-connected.
\end{lem}

\begin{proof}[Proof of Lemma~\ref{lem:sl_xj_conn}]
Choose two arbitrary points $x^{4N}=(x_1,\ldots,x_N) \in \sS_{j,x}$ and $y^{4N}=(y_1,\ldots,y_N) \in \sS_{j,x}$ (in this notation $x_j=x$, $y_j=x$). We construct a path from $x^{4N}$ to $y^{4N}$ by a similar procedure that was used in the proof of Lemma~\ref{lem:sl_conn}. First, move the time variables $x^0_1,\ldots,x^0_N$ to the common time $x^0_1,\ldots,x^0_N=x^0$ along the path
\be
\gamma:[0,1]\to \RRR^{4N}, \gamma(s)=\Bigl( \bigl( (1-s)x^0_1+s x^0,\vx_1 \bigr),\ldots, \bigl( (1-s)x^0_N+s x^0,\vx_N \bigr) \Bigr)\,,
\ee
which stays in $\sS_{j,x}$. Do the same with $y^0_1,\ldots,y^0_N$. Now one can easily move $\vx_1$ to $\vy_1$ by any path that avoids the points $\vx_2,\ldots,\vx_N$ as in the proof of Lemma~\ref{lem:sl_conn}. Repeating that for $\vx_2,\ldots,\vx_N$, one obtains a path connecting $x^{4N}$ and $y^{4N}$.
\end{proof}

\begin{proof}[Proof of Theorem~\ref{thm:inconsistentV4}]
We first prove the generalization of Theorem~\ref{thm:inconsistentV2} and later that of Theorem~\ref{thm:inconsistentV3}. That is, we first assume that the consistency condition holds everywhere, and later show that it suffices to assume the consistency condition on the collision-free spacelike configurations.

We use the notation $x^{4N}=(x_1,\ldots,x_N)$ and $A_{i,0} = I$. Since the $H_i^\free$ commute with each other, and the $V_i$ do, and the $A_{i,1}$, $A_{i,2}$, $A_{i,3}$ and $B_i$ are functions only of $x_i$, the consistency condition \eqref{consistency} is
\begin{align}\label{cons_prop}
0 &= \biggl[-i\sum_{\mu=0}^3 A_{i,\mu}(x_i)\frac{\partial}{\partial x_{i,\mu}} + B_i(x_i) + V_i(x^{4N}), 
-i\sum_{\mu=0}^3 A_{j,\mu}(x_j)\frac{\partial}{\partial x_{j,\mu}} + B_j(x_j) + V_j(x^{4N})\biggr] \nonumber \\
&= \biggl[-i\sum_{\mu=0}^3 A_{i,\mu}(x_i)\frac{\partial}{\partial x_{i,\mu}},V_j(x^{4N})\biggr] 
- \biggl[-i\sum_{\mu=0}^3 A_{j,\mu}(x_j)\frac{\partial}{\partial x_{j,\mu}},V_i(x^{4N})\biggr] \nonumber \\
&= -i\sum_{\mu=0}^3 A_{i,\mu}(x_i)\frac{\partial V_j(x^{4N})}{\partial x_{i,\mu}} 
+i\sum_{\mu=0}^3 A_{j,\mu}(x_j)\frac{\partial V_i(x^{4N})}{\partial x_{j,\mu}}.
\end{align}
We define $k_i \times k_i$ matrices $A_{i,\alpha}(x_i)$ (depending smoothly on $x_i$) for $\alpha = 4,\ldots,k_i^2-1$ such that $\{ A_{i,\alpha}(x_i) \}_{\alpha=0,\ldots,k_i^2-1}$ is a basis (over the field $\RRR$) in the space of self-adjoint $k_i \times k_i$ matrices. Then, by expanding $V_i$ in this basis, i.e.,
\be
V_i(x^{4N}) = \sum_{\alpha=0}^{k_i^2-1} A_{i,\alpha}(x_i) \, d_{i,\alpha}(x^{4N})
\ee
with real-valued smooth functions $d_{i,\alpha}$, the consistency condition \eqref{cons_prop} is
\be
-i\sum_{\mu=0}^3 \sum_{\alpha=0}^{k_j^2-1} A_{i,\mu}(x_i) A_{j,\alpha}(x_j) \frac{\partial d_{j,\alpha}(x^{4N})}{\partial x_{i,\mu}} 
+i \sum_{\mu=0}^3 \sum_{\alpha=0}^{k_i^2-1} A_{j,\mu}(x_j) A_{i,\alpha}(x_i) \frac{\partial d_{i,\alpha}(x^{4N})}{\partial x_{j,\mu}} = 0.
\ee
Since, for $i\neq j$, $A_{i,\alpha}$ and $A_{j,\alpha}$ act on different indices (viz., $s_i$ and $s_j$), we have that, for every $x_i$ and $x_j$, the matrices
\be\label{sevenmatrices3}
\Bigl\{I,A_{i,\alpha}(x_i),A_{j,\beta}(x_j), A_{i,\alpha}(x_i)\,A_{j,\beta}(x_j):\alpha=1\ldots k_i^2-1,\beta=1\ldots k_j^2-1\Bigr\}
\ee
are linearly independent (over the field $\RRR$) when regarded as $k \times k$ matrices ($k=\prod_{\ell=1}^N k_\ell$) acting on all spin indices $s_1,\ldots,s_N$. Thus,
\be\label{relation_prop}
\frac{\partial d_{j,\alpha}}{\partial x_{i,\mu}} = 0 \quad \forall i \neq j,~ \mu = 0,1,2,3,~ \alpha = 4,\ldots,k_j^2-1
\ee
and
\be\label{relation_prop2}
\frac{\partial d_{j,\mu}}{\partial x_{i,\nu}} - \frac{\partial d_{i,\nu}}{\partial x_{j,\mu}} = 0 \quad \forall i \neq j,~ \mu,\nu = 0,1,2,3.
\ee
Equation \eqref{relation_prop} implies that
\be\label{djalphaxj}
d_{j,\alpha}(x_1,\ldots,x_N) = d_{j,\alpha}(x_j)\quad \forall \alpha = 4,\ldots,k_j^2-1\,.
\ee
Now we draw inferences from \eqref{relation_prop2}. For
\be\label{gimunudef}
g_{i,\mu\nu} := \frac{\partial d_{i,\mu}}{\partial x_{i,\nu}} - \frac{\partial d_{i,\nu}}{\partial x_{i,\mu}}
\ee
we find, using \eqref{relation_prop2} twice, that for all $j \neq i$ and $\lambda=0,1,2,3$
\begin{align}\label{relation_prop3}
\frac{\partial}{\partial x_{j,\lambda}} g_{i,\mu\nu}
&= \frac{\partial}{\partial x_{j,\lambda}}\left( \frac{\partial d_{i,\mu}}{\partial x_{i,\nu}} - \frac{\partial d_{i,\nu}}{\partial x_{i,\mu}} \right) \nonumber \\
&= \frac{\partial}{\partial x_{i,\nu}} \frac{\partial d_{i,\mu}}{\partial x_{j,\lambda}} - \frac{\partial}{\partial x_{i,\mu}} \frac{\partial d_{i,\nu}}{\partial x_{j,\lambda}} \nonumber \\
&= \frac{\partial}{\partial x_{i,\nu}} \frac{\partial d_{j,\lambda}}{\partial x_{i,\mu}} - \frac{\partial}{\partial x_{i,\mu}} \frac{\partial d_{j,\lambda}}{\partial x_{i,\nu}} \nonumber \\
&= 0.
\end{align}
This implies that $g_{i,\mu\nu}(x_1,\ldots,x_N)=g_{i,\mu\nu}(x_i)$ is a function of $x_i$ only. By choosing some fixed values $\tilde{x}_1,\ldots,\tilde{x}_N$ we define the function
\be\label{tildeddef}
\tilde{d}_{i,\mu}(x_i) := d_{i,\mu}(\tilde{x}_1,\ldots,\tilde{x}_{i-1},x_i,\tilde{x}_{i+1},\ldots,\tilde{x}_N)\,.
\ee
It has the property that
\be\label{g_prop2}
g_{i,\mu\nu}(x_i)= \left( \frac{\partial \tilde{d}_{i,\mu}}{\partial x_{i,\nu}} - \frac{\partial \tilde{d}_{i,\nu}}{\partial x_{i,\mu}} \right) (x_i).
\ee
Now define $h_{i,\mu} := d_{i,\mu}(x_1,\ldots,x_N)-\tilde{d}_{i,\mu}(x_i)$. Using \eqref{g_prop2} we find
\be\label{h_i}
\frac{\partial h_{i,\mu}}{\partial x_{i,\nu}} - \frac{\partial h_{i,\nu}}{\partial x_{i,\mu}} = 0
\ee
for all $i=1,\ldots,N$ and all $\mu,\nu=0,1,2,3$. Since from \eqref{relation_prop2} we also have that
\be\label{h_ij}
\frac{\partial h_{i,\mu}}{\partial x_{j,\nu}} - \frac{\partial h_{j,\nu}}{\partial x_{i,\mu}} = 0
\ee
for all $i,j=1,\ldots,N$ ($i \neq j$) and all $\mu,\nu=0,1,2,3$, it follows that $h_{i,\mu}$ is a gradient of a real-valued function $\theta$, i.e., $h_{i,\mu}=\frac{\partial \theta}{\partial x_{i,\mu}}$. Therefore
\be
d_{j,\mu}(x^{4N}) = \frac{\partial \theta(x^{4N})}{\partial x_{j,\mu}} + \tilde{d}_{j,\mu}(x_j).
\ee
By \eqref{djalphaxj},
\be\label{Vjgauge}
V_j(x^{4N}) = \sum_{\mu=0}^3 A_{j,\mu}(x_j) \frac{\partial \theta(x^{4N})}{\partial x_{j,\mu}} + \underbrace{\sum_{\mu=0}^3 A_{j,\mu}(x_j)\tilde{d}_{j,\mu}(x_j) + \sum_{\alpha=4}^{k_j^2-1} A_{j,\alpha}(x_j)d_{j,\alpha}(x_j)}_{\tilde{V}_{j}(x_j)}\,,
\ee
and \eqref{inconVfree} follows.

Now we discuss how this proof needs to be changed when the consistency condition is granted only on $\sS_{\neq}$. In this case \eqref{relation_prop} and \eqref{relation_prop2} hold only on $\sS_{\neq}$. From \eqref{relation_prop} on $\sS_{\neq}$ we can still conclude \eqref{djalphaxj} because $d_{j,\alpha}$ has vanishing gradient on the set $\sS_{j,x_j}$, which is path-connected according to Lemma~\ref{lem:sl_xj_conn}. For the same reason also \eqref{relation_prop3}, valid on $\sS_{\neq}$, implies that $g_{i,\mu\nu}$ is a function of $x_i$ only. Since every $x_i\in\RRR^4$ occurs in some spacelike configuration, we obtain that $g_{i,\mu\nu}(x_i)$ is, in fact, consistently (and smoothly) defined on all $\RRR^4$; that is, for any fixed $i$, writing $x$ for $x_i$ and $g_{\mu\nu}(x)$ for $g_{i,\mu\nu}(x_i)$, we have a 2-form $g$ on $\RRR^4$ that is closed by virtue of \eqref{gimunudef}. By the Poincar\'e lemma, it is exact, i.e., $g$ is the exterior derivative of a 1-form $\tilde{d}$ on $\RRR^4$; now write $\tilde{d}_{i,\mu}(x_i)$ instead of $\tilde{d}_\mu(x)$ and take this, rather than \eqref{tildeddef}, as the definition of $\tilde{d}$. Then \eqref{g_prop2} is still true, both when regarded as a relation between functions on $\RRR^4$ or between functions on $\sS_{\neq}$. It now follows that \eqref{h_i} and \eqref{h_ij} are still true on $\sS_{\neq}$. By the Poincar\'{e} lemma again, a closed 1-form on a simply connected set is exact,\footnote{Note that the Poincar\'{e} lemma for 1-forms indeed holds on simply connected sets. For $p$-forms with $p \geq 2$ one needs stronger assumptions, e.g., that the $p$-form is closed on a contractible set (such as a star shaped set), in order to conclude that it is exact.} and since, according to Lemma~\ref{lem:sl_conn}, $\sS_{\neq}$ is simply connected, it follows that $h_{i\mu}$ is the gradient of a scalar function $\theta:\sS_{\neq}\to\RRR$. Thus \eqref{Vjgauge} is still true and \eqref{inconVfree} follows on $\sS_{\neq}$.
\end{proof}

\begin{proof}[Proof of Theorem~\ref{thm:inconsistent2ndorder}]
We use the convention $A_{i,ab}(x_i) = A_{i,ba}(x_i)$ for all $i$. Then the consistency condition is that, for an arbitrary initial wave function $\psi$,
\begin{align}
0 &= \biggl[ i\frac{\partial}{\partial t_i} - H_i , i\frac{\partial}{\partial t_j} - H_j \biggr]\psi \nonumber \\
&= \sum_{a,b=1}^3 \left( A_{i,ab}(x_i) \frac{\partial^2 V_j}{\partial x_{i,a}\partial x_{i,b}} - A_{j,ab}(x_j) \frac{\partial^2 V_i}{\partial x_{j,a}\partial x_{j,b}} \right) \psi \nonumber \\
&\quad + \sum_{a=1}^3 \left( B_{i,a}(x_i) \frac{\partial V_j}{\partial x_{i,a}} - B_{j,a}(x_j) \frac{\partial V_i}{\partial x_{j,a}} \right) \psi + i \left( \frac{\partial V_i}{\partial t_j} - \frac{\partial V_j}{\partial t_i} \right) \psi \nonumber \\
&\quad + 2\sum_{a,b=1}^3 A_{i,ab}(x_i) \frac{\partial V_j}{\partial x_{i,a}} \frac{\partial \psi}{\partial x_{i,b}} - 2\sum_{a,b=1}^3 A_{j,ab}(x_j) \frac{\partial V_i}{\partial x_{j,a}} \frac{\partial \psi}{\partial x_{j,b}}.
\end{align}
Since $\psi$ is arbitrary, the terms involving $\partial\psi/\partial x_{i,b}$, $\partial \psi/\partial x_{j,b}$, and $\psi$ must vanish separately; that is,
\begin{align}
0&=\sum_{a=1}^3 A_{i,ab}(x_i) \frac{\partial V_j}{\partial x_{i,a}} \label{diaVj}\\
0&=\sum_{a=1}^3 A_{j,ab}(x_j) \frac{\partial V_i}{\partial x_{j,a}}\\
0&= \sum_{a,b=1}^3 \left( A_{i,ab}(x_i) \frac{\partial^2 V_j}{\partial x_{i,a}\partial x_{i,b}} 
- A_{j,ab}(x_j) \frac{\partial^2 V_i}{\partial x_{j,a}\partial x_{j,b}} \right)  \nonumber \\
&\quad + \sum_{a=1}^3 \left( B_{i,a}(x_i) \frac{\partial V_j}{\partial x_{i,a}} - B_{j,a}(x_j) \frac{\partial V_i}{\partial x_{j,a}} \right) + i \left( \frac{\partial V_i}{\partial t_j} - \frac{\partial V_j}{\partial t_i} \right) \label{ddVj}\,.
\end{align}
In \eqref{diaVj}, we can replace $A_{i,ab}$ by $A_{i,ba}$; since the $3k_i\times 3k_i$ matrix $A_i=(A_{i,ba})$ has full rank, it possesses an inverse $A_i^{-1}$; multiplying \eqref{diaVj} from the left by $A_i^{-1}$, we obtain that
\be\label{diaVj2}
\frac{\partial V_j}{\partial x_{i,a}}=0\,,
\ee
as well as the same relation with $i$ and $j$ interchanged. Thus, all spatial derivatives of $V$ in \eqref{ddVj} drop out, leaving us with
\be\label{dtiVj2}
 \frac{\partial V_i}{\partial t_j} - \frac{\partial V_j}{\partial t_i} =0\,.
\ee
Since \eqref{diaVj2} and \eqref{dtiVj2} coincide with \eqref{relation1} and \eqref{relation2}, the proof of Theorem~\ref{thm:inconsistentV2} goes through from \eqref{tilde_relation} onwards, with the only difference that we have not presented any reason why $W_j$, or $\theta$, should be real-valued; rather, the arguments in the proof of Theorem~\ref{thm:inconsistentV2} yield, in our case, the existence of a self-adjoint matrix-valued $\theta$. However, since $V_i$ was assumed to act only on $s_i$, it follows that also $\tilde{V}_i$, $W_i$, and $\theta$ act only on $s_i$; since $i$ was arbitrary, $\theta$ must be scalar (and thus real-valued). 

In order to verify \eqref{inconVfree}, it is relevant to note that $\theta$ depends on $x_1,\ldots,x_N$ only through $t_1,\ldots,t_N$, while $H_i^\free$ involves only spatial derivatives; thus, no magnetic terms arise from the gauge transformation, and $\tilde{\phi}$ satisfies \eqref{inconVfree}.
\end{proof}

\section{Proof of Cut-Off Consistency Theorem}
\label{sec:proof_delta}

In this section we prove Theorem~\ref{thm:delta_model} of Section~\ref{sec:delta_model}. Before we begin the proof, we need a few preparatory considerations. 
For later reference, let us write out the one-particle Dirac equation with potential for a wave function $\psi: \RRR^4 \to \CCC^4$,
\begin{equation}\label{one_part_Dirac_pot}
i\frac{\partial}{\partial t} \psi(t, \vect{x}) = \Bigl(-i \vect{\alpha} \cdot \nabla + \beta m + V(\vect{x}) \Bigr) \psi(t, \vect{x})
\end{equation}
with positive constant $m$ and a self-adjoint $4 \times 4$-matrix-valued function $V$.
The $N$-particle single-time Dirac equation with potential for a wave function $\psi: \RRR^{3N+1} \to (\CCC^4)^{\otimes N}$ is
\begin{equation}\label{N_part_Dirac_pot}
i\frac{\partial}{\partial t} \psi(t, q) = \left( \sum_{k=1}^{N} \left(-i \vect{\alpha}_k \cdot \nabla_k + \beta_k m \right) + V(q) \right) \psi(t, q)
\end{equation}
with a self-adjoint matrix-valued function $V: \RRR^{3N} \to (\CCC^{4 \times 4})^{\otimes N}$. 

Since Theorem~\ref{thm:delta_model} is formulated in terms of \emph{smooth} functions, we need to make use of known results on the regularity of solutions to the Dirac equation; for this purpose we quote in Lemma~\ref{lemma_chernoff} a basic result from \cite{chernoff:1973}. 
By a \emph{Dirac-type differential operator} we mean an operator
of the form
\begin{equation}\label{def_dirac_type_operator}
H\psi(q)
= -i\sum_{i=1}^{d} A_i(q) \frac{\partial \psi(q)}{\partial q_i} + B(q) \psi(q)\,,
\end{equation}
where the coefficients $A_i(q)$ and $B(q)$ are self-adjoint operators on
$\CCC^k$ and smooth functions of $q\in\RRR^d$.

\begin{lem}[\cite{chernoff:1973}]\label{lemma_chernoff}
If $H$ is a Dirac-type differential operator and if there exists a constant $c>0$ such that $\|\sum_i n_i A_i(q)\|\leq c$ for all $q\in\RRR^d$ and all unit vectors $(n_1,\ldots, n_d)$, then the PDE
\begin{equation}\label{Chernoffb}
i\frac{\partial f(t,q)}{\partial t}
= Hf(t,q),
\end{equation}
where $t\in\RRR$, $q\in \RRR^d$, and $f$ is $\CCC^k$-valued,
possesses, for any smooth initial datum $f_0\in C^\infty(\RRR^d,\CCC^k)$, a
unique solution $f\in C^\infty(\RRR^{d+1},\CCC^k)$ with $f(0,\cdot)=f_0(\cdot)$.
\end{lem}

The constant $c$ can be taken to be 1 in the 1-particle Dirac equation \eqref{one_part_Dirac_pot} and $\sqrt{N}$ in the many-particle Dirac equation \eqref{N_part_Dirac_pot}.

\subsection{Domain of Dependence}\label{finite_speed_N}

It is well known that, in the Dirac equation, perturbations propagate no faster than at the speed of light (here, $c=1$). Lemma~\ref{finite_propagation_speed_N_part_dirac} below is the appropriate version of that statement for the $N$-particle Dirac equation \eqref{N_part_Dirac_pot} and asserts that the domain of dependence of $\psi(t,\vx_1,\ldots,\vx_N)$ at time 0 in $\RRR^{3N}$ is the Cartesian product of $N$ 3-balls of radius $|t|$. To make its proof easier to follow, we first give an analogous proof of the corresponding well-known statement for $N=1$, formulated below as Lemma~\ref{finite_propagation_speed_one_part_dirac}.

We denote the closed ball around $\vx$ with radius $r \geq 0$ by
\be\label{Bdef}
\overline B_{r}(\vx) = \bigl\{ \vy \in \RRR^3: \|\vy-\vx\| \leq r \bigr\}.
\ee
The ball around the origin is abbreviated by $\overline B_{r} := \overline B_{r}(\vect{0})$. The ball around $\vx$ can also be written as $\overline B_{r}(\vx) = \vx + \overline B_{r}$ with the summation defined as $a+B := \{ a+b: b \in B \}$ for $a \in \RRR^d$ and $B \subset \RRR^d$. 
The next Lemma~\ref{finite_propagation_speed_one_part_dirac} states that the domain of dependence for $\psi(t,\vect{x})$ is given by $\overline B_{|t|}(\vect{x})$.

\begin{lem}
\label{finite_propagation_speed_one_part_dirac}
Let $\psi$ be a solution of the one-particle Dirac equation \eqref{one_part_Dirac_pot} with initial data $\psi(0,\cdot) \in C^{\infty}(\RRR^3, \CCC^4)$ and with a self-adjoint potential matrix $V \in C^{\infty}(\RRR^3, \CCC^{4 \times 4})$. 
Then specifying initial conditions on $\overline B_{|t|}(\vect{x})$ uniquely determines $\psi(t,\vect{x})$ (i.e., the domain of dependence is a ball growing at the speed of light).
\end{lem}

\begin{proof}
For simplicity, we consider only $t\geq 0$. Let $(T, \vect{y}) \in \RRR^4$ with $T>0$, and for all $t \in [0,T]$ let
\begin{equation}
\Sigma_t = \left\{ (t,\vect{x}) \in \RRR^4: \vect{x} \in \overline B_{T-t}(\vect{y}) \right\}\,.
\end{equation}
We first prove that if $\psi$ vanishes on $\Sigma_0$ it also vanishes on $\Sigma_t$ for all $t \in [0,T]$.

\begin{figure}[htbp]
\centering
\includegraphics[width=280pt,keepaspectratio]{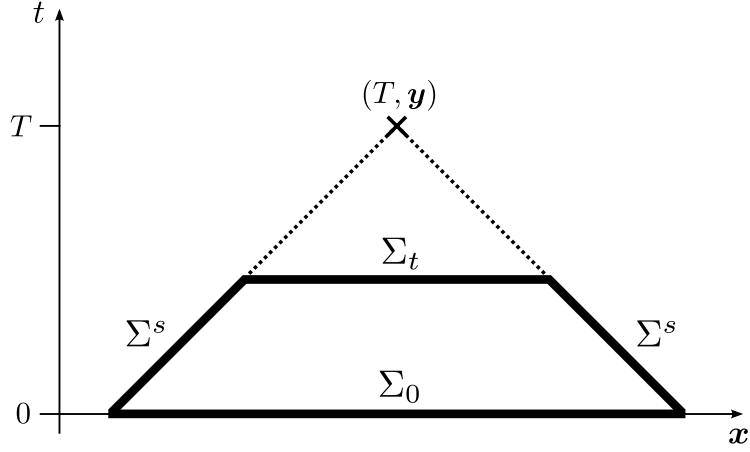}
\centering
\caption{\small{Truncated light cone pertaining to the tip $(T,\vy)$ and enclosed by the faces $\Sigma_0,\Sigma_t$, and $\Sigma^s$. $\Sigma_0$ and $\Sigma_t$ are 3-balls at constant coordinate time, $\Sigma^s$ is lightlike.}}
\label{figure:finite_speed_proof}
\end{figure}

For any $t \in [0,T)$, let the light cone between the surfaces $\Sigma_0$ and $\Sigma_t$ be $C_t = \bigcup_{t' \in [0,t]} \Sigma_{t'}$. Let $\Sigma^s$ denote the sides of the light cone, i.e., $\Sigma^s = \partial C_t \setminus (\Sigma_0\cup \Sigma_t)$ (with $\partial C_t$ the boundary of $C_t$), as shown in Figure~\ref{figure:finite_speed_proof}. For ease of notation and since we consider a fixed time $t$ we do not write an index $t$ for the $t$-dependence of $\Sigma^s$, and we write $C := C_t$. Let $n$ be the outward-pointing unit vector field on $\partial C$ orthogonal to $\partial C$ in the Euclidean metric on $\RRR^4$, i.e., $\|n\|^2 = \sum_{i=1}^4 n^i n^i = 1$ and for any tangent vector $s$ on $\partial C$, $n \cdot s = \sum_{i=1}^4 n^i s^i = 0$.\footnote{We use a normal vector in the Euclidean metric here because we have to deal with the lightlike hypersurface $\Sigma^s$, so the normal vector in the Minkowski sense would be lightlike, too. Therefore the flux integrals \eqref{flux_int} could not be written down in this simple form.} 
Let $\psi^{\dagger}$ denote the conjugate-transpose of the spinor $\psi$ (so that $|\psi|^2 = \psi^{\dagger}\psi$). Let $j = (j^0,j^1,j^2,j^3) = (|\psi|^2, \psi^{\dagger}\alpha^1\psi, \psi^{\dagger}\alpha^2\psi, \psi^{\dagger}\alpha^3\psi)$ denote the four-current. Then the continuity equation can be expressed as $\divergence(j) = 0$.\footnote{Note the difference between the four-divergence $\divergence(j) = \partial_{\mu}j^{\mu} = \frac{\partial j^0}{\partial t} + \nabla \cdot \vect{j}$ and the three-divergence $\divergence(\vect{j}) = \nabla \cdot \vect{j}$.} According to the Gauss integral theorem in 4 dimensions,
\begin{equation}\label{flux_int}
0 = \int_C \divergence(j)~d^4x = \int_{\partial C} j \cdot n ~d^3x = \int_{\Sigma_0} j \cdot n ~d^3x + \int_{\Sigma_t} j \cdot n ~d^3x + \int_{\Sigma^s} j \cdot n ~d^3x.
\end{equation}
The differential $d^3x$ denotes the volume on a 3-surface relative to the Euclidean metric and $j \cdot n = \sum_{k=1}^3 j^k n^k$ denotes the Euclidean inner product.
Now suppose that $\psi|_{\Sigma_0} = 0$. Since the (outward-pointing) normal vector on $\Sigma_0$ is $n = (-1,0,0,0)$ we have $j \cdot n = -j^0 = -|\psi|^2 = 0$ on $\Sigma_0$. On $\Sigma_t$ the normal vector is $n = (1,0,0,0)$, so $j \cdot n = j^0 = |\psi|^2$ on $\Sigma_t$. Therefore, 
\begin{equation}
0 = \int_{\Sigma_t} j \cdot n ~d^3x + \int_{\Sigma^s} j \cdot n ~d^3x = \int_{\Sigma_t} |\psi|^2 ~d^3x + \int_{\Sigma^s} \psi^{\dagger} (n^0 I + \vect{\alpha} \cdot \vect{n}) \psi ~d^3x.
\end{equation}
Next we prove that the $4 \times 4$ matrix $A := (n^0 I + \vect{\alpha} \cdot \vect{n})$ is positive semi-definite on $\Sigma^s$, i.e., that all eigenvalues are $\geq 0$. For any unit vector $\vect{b} \in \RRR^3$ we have that $\vect{\alpha} \cdot \vect{b}$ has eigenvalues $-1$ and $+1$.\footnote{\label{footnote:Dirac_matrices}Indeed, for the Dirac matrices the following relation holds: $\alpha^i\alpha^j + \alpha^j\alpha^i = 2\delta_{ij} I$. Therefore, for any unit vector $\vect{b} \in \RRR^3$ we have $(\vect{\alpha} \cdot \vect{b})^2 = \sum_{i,j=1}^3 b^i b^j \alpha^i\alpha^j = \sum_{i,j=1}^3 b^i b^j (2\delta_{ij} I - \alpha^j\alpha^i) = -(\vect{\alpha} \cdot \vect{b})^2 + 2\|\vect{b}\|^2 I$, i.e., $(\vect{\alpha} \cdot \vect{b})^2 = \|\vect{b}\|^2 I = I$, so $\vect{\alpha} \cdot \vect{b}$ has eigenvalues $-1$ and $+1$.} Therefore the matrix $\|\vect{n}\| \vect{\alpha} \cdot \frac{\vect{n}}{\|\vect{n}\|}$ has eigenvalues $+\|\vect{n}\|$ and $-\|\vect{n}\|$, so the lowest eigenvalue of $A$ is $e = n^0 - \|\vect{n}\|$. On $\Sigma^s$ the normal vector is $n = (\frac{1}{\sqrt{2}}, \vect{n})$ with $\|\vect{n}\| = \frac{1}{\sqrt{2}}$, so $e = \frac{1}{\sqrt{2}} - \frac{1}{\sqrt{2}} = 0$. We thus have that
\begin{equation}\label{int_Sigma_tSigma^s}
0 = \int_{\Sigma_t} |\psi|^2 ~d^3x + \int_{\Sigma^s} \underbrace{\psi^{\dagger} (n^0 I + \vect{\alpha} \cdot \vect{n}) \psi}_{\geq 0} ~d^3x.
\end{equation}
Since each integrand is $\geq 0$, each integral is $\geq 0$, and thus each integral has to vanish. In particular this means that the integrand $|\psi|^2$ has to vanish (almost everywhere) on $\Sigma_t$, i.e., $\psi|_{\Sigma_t} = 0$ (almost everywhere). From Lemma~\ref{lemma_chernoff} we have that $\psi \in C^{\infty}(\RRR^4, \CCC^4)$, therefore $\psi$ vanishes identically on $\Sigma_t$ for all $t \in [0,T]$ (in particular also for $t=T$). Thus, if $\psi$ vanishes on $\Sigma_0$, then it also vanishes on $\Sigma_t$ for all $t \in [0,T]$. 

The statement of the lemma follows in this way:
Suppose $\psi_1, \psi_2 \in C^{\infty}(\RRR^4, \CCC^4)$ are solutions of the Dirac equation \eqref{one_part_Dirac_pot} and are identical on $\Sigma_0$ and arbitrary on the rest of the $t = 0$ hypersurface. Then $\psi_1 - \psi_2$ is also a solution of the Dirac equation with $\psi_1 - \psi_2 = 0$ on $\Sigma_0$. Therefore, also $\psi_1 - \psi_2 = 0$ in $(T,\vect{y})$, i.e., $\psi_1(T,\vect{y}) = \psi_2(T,\vect{y})$.
\end{proof}

From Lemma~\ref{finite_propagation_speed_one_part_dirac} it follows immediately that the wave function on any $M_0 \subset \RRR^3$ at time $t$ is uniquely determined by specifying initial conditions on $M_{|t|} = M_0 + \overline B_{|t|}$ at time zero (with the summation of sets defined as $A+B := \{ a+b: a \in A, b \in B \}$; pictorally speaking, one obtains $M_{|t|}$ by putting a ball with radius $|t|$ around each point of $M_0$). That is, the domain of dependence for $M_0$ is given by $M_{|t|}$. It also follows that if the initial wave function $\psi(0,\cdot)$ has compact support $M_0\subset\RRR^3$, then $\psi(t,\cdot)$ has compact support in $M_0+\overline B_{|t|}$.

We now generalize these considerations to the $N$-particle Dirac equation. We first generalize the notion of a set growing at the speed of light. Consider a point $q=(\vx_1,\ldots,\vx_N) \in \RRR^{3N}$ and, around each point, a ball growing at the speed of light. The corresponding set in configuration space is
\begin{equation}
\overline B_{t}^{(N)}(q) := 
\prod_{i=1}^N \overline B_{t}(\vect{x}_i).
\end{equation}
This set can also be expressed in terms of the $\|\cdot\|_{2,\infty}$-norm, defined for $q \in \RRR^{3N}$ by
\begin{equation}
\|q\|_{2,\infty} := \max_{i=1,\dots,N} \|\vx_i\|\,,
\end{equation}
according to
\be
\overline B_{t}^{(N)}(q) = \bigl\{ p \in \RRR^{3N}: \|p-q\|_{2,\infty} \leq t \bigr\}\,.
\ee
Indeed, using the notation $\tilde q = (\tilde\vx_1, \dots, \tilde\vx_N)$,
\begin{align} 
\prod_{i=1}^N \overline B_{t}(\vect{x}_i) 
&= \{ \tilde q \in \RRR^{3N}: \tilde\vx_i \in \overline B_{t}(\vect{x}_i) ~\forall i=1,\dots,N \} \nonumber \\ 
&= \{ \tilde q \in \RRR^{3N}: \|\tilde\vx_i - \vect{x}_i\| \leq t ~\forall i=1,\dots,N \} \nonumber \\ 
&= \{ \tilde q \in \RRR^{3N}: \max_{i=1,\dots,N} \|\tilde\vx_i - \vect{x}_i\| \leq t \} \nonumber \\ 
&= \{ \tilde q \in \RRR^{3N}: \|\tilde q - q\|_{2,\infty} \leq t \}\,. 
\end{align}
For the set around the origin we also write $\overline B_{t}^{(N)}(0) = (\overline B_{t}(0))^N =: \overline B_{t}^{(N)} $. We then have $\overline B_{t}^{(N)}(q) = q + \overline B_{t}^{(N)}$. If an arbitrary subset of configuration space $M_0 \subset \RRR^{3N}$ grows with the speed of light, then one obtains $M_{|t|} = M_0 + \overline B_{|t|}^{(N)}$.

\begin{lem}
\label{finite_propagation_speed_N_part_dirac}
Let $\psi$ be a solution of the $N$-particle Dirac equation \eqref{N_part_Dirac_pot} with initial data $\psi(0,\cdot) \in  C^{\infty}(\RRR^{3N}, (\CCC^4)^{\otimes N})$ and with a self-adjoint potential matrix $V \in C^{\infty}(\RRR^{3N}, (\CCC^{4 \times 4})^{\otimes N})$. 
Then specifying initial conditions on $\overline B_{|t|}^{(N)}(q)$ uniquely determines $\psi(t,q)$.
\end{lem}

\begin{proof}
The proof proceeds along the same lines as in the one-particle case. Again, we take $t\geq 0$. We first generalize the definition of $\Sigma_t$ to the $N$-particle case. Let $Q = (\vect{Q}_1, \dots, \vect{Q}_N) \in \RRR^{3N}$ and $(T,Q) \in \RRR^{3N+1}$ with $T>0$. 
For $t \in [0,T]$, let
\begin{equation}
\Sigma_t = \left\{ (t,q) \in \RRR^{3N+1}: q \in \overline B_{T-t}^{(N)}(Q) \right\}\,.
\end{equation}
We first prove that if $\psi$ vanishes on $\Sigma_0$ it also vanishes on $\Sigma_t$ for all $t \in [0,T]$.

We define $C_t = \bigcup_{t' \in [0,t]} \Sigma_{t'}$ as the ``generalized light cone'' between the surfaces $\Sigma_0$ and $\Sigma_t$. Again, $\Sigma^s$ denotes the sides of the generalized light cone, i.e., $\Sigma^s=\partial C_t \setminus (\Sigma_0 \cup \Sigma_t)$. For ease of notation, and since we consider a fixed time $t$, we do not make explicit the $t$-dependence of $\Sigma^s$ and we write $C$ instead of $C_t$. $\Sigma^s$ is composed of $N$ faces in the following sense. We call
\begin{equation}
\Sigma^{s,k} = \bigcup_{t' \in (0,t)} \left\{ (t',q) \in \Sigma^s: \|\vect{Q}_k - \vect{x}_k\| = T-t' \right\}
\end{equation}
the $k$-th face of $\Sigma^s$. Then we have that $\Sigma^s = \bigcup_{k=1,\dots,N} \Sigma^{s,k}$. Now let $n = (n^0,\vect{n}_1,\dots,\vect{n}_N) \in \RRR^{3N+1}$ be the outward-poining unit vector field on $\partial C$ orthogonal to $\partial C$ in the Euclidean metric, i.e., $\|n\|^2 = \sum_{i=1}^{3N+1} n^i n^i = 1$ and for any tangent vector $s$ on $\partial C$, $n \cdot s = \sum_{i=1}^{3N+1} n^i s^i = 0$. The current $j$ for $N$ particles is $j = (j^0,\vect{j}_1,\dots,\vect{j}_N) = (|\psi|^2, \psi^{\dagger}\vect{\alpha}_1\psi, \dots, \psi^{\dagger}\vect{\alpha}_N\psi)$. The well-known continuity equation for $N$ particles reads
\be
\divergence(j) = \frac{\partial |\psi|^2}{\partial t} +\sum_{k=1}^N \divergence(\vect{j}_k) = 0\,.
\ee
According to the Gauss integral theorem in $3N+1$ dimensions,
\begin{equation}
0 = \int_{C} \divergence(j)~d^{3N+1}x = \int_{\partial C} j \cdot n ~d^{3N}x = \int_{\Sigma_0} j \cdot n ~d^{3N}x + \int_{\Sigma_t} j \cdot n ~d^{3N}x + \int_{\Sigma^s} j \cdot n ~d^{3N}x.
\end{equation}
The differential $d^{3N}x$ denotes the $3N$-dimensional surface area relative to the Euclidean metric on $\RRR^{3N+1}$, and $j \cdot n = \sum_{k=1}^{3N+1} j^k n^k$ is the Euclidean inner product on $\RRR^{3N+1}$. We suppose $\psi|_{\Sigma_0} = 0$. As the normal vector on $\Sigma_0$ has components $n^0 = -1$ and $\vect{n}_k = \vect{0}$ (for all $k=1,\dots,N$), it follows that $j \cdot n = -j^0 = -|\psi|^2 = 0$ on $\Sigma_0$. On $\Sigma_t$ we have $n^0 = 1$ and $\vect{n}_k = \vect{0}$ (for all $k=1,\dots,N$), so $j \cdot n = j^0 = |\psi|^2$. Therefore, 
\begin{equation}
0 = \int_{\Sigma_t} |\psi|^2 ~d^{3N}x + \int_{\Sigma^s} \psi^{\dagger} \Bigl(n^0 I + \sum_{k=1}^N \vect{\alpha}_k \cdot \vect{n}_k\Bigr) \psi ~d^{3N}x.
\end{equation}
The $4^N \times 4^N$ matrix $A := (n^0 I + \sum_{k=1}^N \vect{\alpha}_k \cdot \vect{n}_k)$ is positive semi-definite on $\Sigma^s$ for the following reason. Since for any unit vector $\vect{b} \in \RRR^3$ we have that $\vect{\alpha} \cdot \vect{b}$ has eigenvalues $-1$ and $+1$ (see Footnote~\ref{footnote:Dirac_matrices}), each matrix $\|\vect{n}_k\| \vect{\alpha}_k \cdot \frac{\vect{n}_k}{\|\vect{n}_k\|}$ has eigenvalues $+\|\vect{n}_k\|$ and $-\|\vect{n}_k\|$. Then the lowest eigenvalue of $A$ is $e = n^0 - \sum_{k=1}^N\|\vect{n}_k\|$. On $\Sigma^s$ the normal-vector has the component $n^0 = \frac{1}{\sqrt{2}}$. The spatial components depend on the face of $\Sigma^s$. At $\Sigma^{s,k}$ the spatial components have norm $\|\vect{n}_j\| = \frac{1}{\sqrt{2}} \delta_{jk}$ (for all $j=1,\dots,N$), so $e = \frac{1}{\sqrt{2}} - \sum_{j=1}^N \frac{1}{\sqrt{2}} \delta_{jk} = 0$. Thus
\begin{equation}\label{int_Sigma_tSigma^s_N}
0 = \int_{\Sigma_t} |\psi|^2 ~d^{3N}x + \int_{\Sigma^s} \underbrace{\psi^{\dagger} \left( n^0 I + \sum_{k=1}^N \vect{\alpha}_k \cdot \vect{n}_k \right) \psi}_{\geq 0} ~d^{3N}x.
\end{equation}
Since each integrand is $\geq 0$, each integral is $\geq 0$, so each integral has to vanish. This means in particular for the integral over $\Sigma_t$, that the integrand $|\psi|^2$ has to vanish (almost everywhere) and therefore $\psi = 0$ on $\Sigma_t$ (almost everywhere). Since from Lemma~\ref{lemma_chernoff} we know that $\psi \in C^{\infty}(\RRR^{3N+1}, (\CCC^4)^{\otimes N})$, $\psi$ has to vanish identically on $\Sigma_t$. Thus, if $\psi$ vanishes on $\Sigma_0$, then it also vanishes on $\Sigma_t$ for all $t \in [0,T]$. 

Now Lemma~\ref{finite_propagation_speed_N_part_dirac} follows: 
As in the one-particle case, suppose $\psi_1, \psi_2 \in C^{\infty}(\RRR^{3N+1}, (\CCC^4)^{\otimes N})$ are solutions of the $N$-particle Dirac equation \eqref{N_part_Dirac_pot} that are identical on $\Sigma_0$ and arbitrary on the rest of the $t = 0$ hypersurface. Then $\psi_1 - \psi_2$ is another solution of \eqref{N_part_Dirac_pot} with $\psi_1 - \psi_2 = 0$ on $\Sigma_0$. Thus, $\psi_1 - \psi_2 = 0$ in $(T,Q)$, i.e., $\psi_1(T,Q) = \psi_2(T,Q)$.
\end{proof}

Lemma~\ref{finite_propagation_speed_N_part_dirac} implies that the wave function at time $t$ on any $M_0 \subset \RRR^{3N}$ is uniquely determined by the initial conditions on $M_{|t|} = M_0 + \overline B_{|t|}^{(N)}$ at time zero, i.e., the domain of dependence for $M_0$ is given by $M_{|t|}$. It also implies that if the initial wave function $\psi(0,\cdot)$ has compact support $M_0\subset \RRR^{3N}$ 
then $\psi(t,\cdot)$ has compact support in $M_0+\overline B_{|t|}^{(N)}$.

\subsection{Proof of the Theorem}

\begin{proof}[Proof of Theorem~\ref{thm:delta_model}]
We will define a smooth function $\Phi$ on $\sS_{\delta}$; show that any smooth solution $\phi$ of the evolution equations with initial conditions $\phi_0$ must agree with $\Phi$; and show that such a solution exists by showing that $\Phi$ satisfies the multi-time equations \eqref{multi_time_Dirac_equation_N}, which we abbreviate as
\be\label{PhiHalpha}
i\frac{\partial \Phi}{\partial t_\alpha} = H_{S_\alpha} \Phi
\ee
with
\be\label{Halphadef}
H_{S_\alpha}=
\sum_{j \in S_\alpha} H_j^\free + \sum_{\substack{i,j \in S_\alpha \\ i \neq j}} W(\vect{x}_i - \vect{x}_j)\,.
\ee
To this end, let $\phi$ be any smooth solution of the multi-time equations \eqref{PhiHalpha} on $\sS_\delta$ with initial conditions $\phi_0$ on $(\{0\}\times\RRR^3)^N$. We proceed by induction on the number $L$ of families in a partition and treat each $\sS_{\delta,P}$ separately.

We start the induction with $L=1$. The corresponding partition is $P = \{ S_1 \}$ with $S_1 = \{ 1,\dots,N \}$. Note that with $P = \{ S_1 \}$ we have $q^4 = (t_1,q_1)$ with $q_1 = (\vect{x}_1,\dots,\vect{x}_N)$. Let $\Phi$ on $\sS_{\delta,\{S_1\}}$ be the solution (which exists, is unique, and is smooth by Lemma~\ref{lemma_chernoff}) of
\begin{equation}\label{ind_eq_1}
i \frac{\partial}{\partial t_1} \Phi(t_1,q_1) = 
H_{\{1,\ldots,N\}} \, \Phi(t_1,q_1)
\end{equation}
with initial conditions given by $\phi_0$. This is just a single-time Dirac-type equation. Any given solution $\phi$ has to agree with $\Phi$ on $\sS_{\delta,\{S_1\}}$ by Lemma~\ref{lemma_chernoff}, since both functions have the same initial conditions $\phi_0$ and satisfy the same equation \eqref{ind_eq_1}. For showing that $\Phi$ satisfies \eqref{PhiHalpha}, the induction start is also provided by \eqref{ind_eq_1}. 

The induction assumption asserts that $\Phi$ has been defined on the union of the $\sS_{\delta,P'}$ for all partitions $P'$ with $L'=L-1$ or fewer families, that it is smooth on each $\sS_{\delta,P'}$, that any smooth solution $\phi$ of \eqref{PhiHalpha} with initial condition $\phi_0$ agrees with $\Phi$ on $\sS_{\delta,P'}$, and that $\Phi$ satisfies \eqref{PhiHalpha} for all $\alpha$ on all $\sS_{\delta,P'}$ with $L'<L$.

Now we carry out the induction step from $L-1$ to $L$. Consider any $P$ consisting of $L$ families. We now define $\Phi$ on $\sS_{\delta,P}$. That is, we construct $\Phi(Q^4)$ for an arbitrary $Q^4 = (T_1,Q_1;\ldots;T_L,Q_L) \in \sS_{\delta,P}$, numbering the families in $P = \{ S_1,\ldots,S_L \}$ so that $T_1\leq T_2\leq \ldots\leq T_L$. 
According to the induction assumption, $\Phi$ is already given on $\sS_{\delta,\tilde{P}}$ for
\be\label{tildePdef}
\tilde{P} = \Bigl\{ S_1,\dots,S_{L-2},S_{L-1} \cup S_L \Bigr\}\,.
\ee
In particular, $\Phi\bigl(T_1,Q_1;\ldots;T_{L-1},Q_{L-1};T_{L-1},q_L\bigr)$ is given for every $q_L \in \overline B_{T_L-T_{L-1}}^{(|S_L|)}(Q_L)$ because every such $(T_1,Q_1;\ldots;T_{L-1},Q_{L-1};T_{L-1},q_L)$ lies in $\sS_\delta$ (using $T_1\leq \ldots \leq T_L$) and thus in $\sS_{\delta,P}\cap \sS_{\delta,\tilde{P}}$.
Since the domain of dependence of $Q_L$ is $\overline B_{T_L-T_{L-1}}^{(|S_L|)}(Q_L)$ by Lemma~\ref{finite_propagation_speed_N_part_dirac}, we can uniquely solve the single-time equation
\be\label{ind_eq_2}
i \frac{\partial}{\partial t_L} \Phi(T_1,Q_1; \dots; T_{L-1},Q_{L-1};t_L,q_L) = 
H_{S_L}\,\Phi(T_1,Q_1; \dots; T_{L-1},Q_{L-1};t_L,q_L)
\ee
in the variables $t_L,q_L$ with initial condition given by $\Phi$ on the set
\be\label{DoDsetL}
\Bigl\{(T_1,Q_1;\dots;T_{L-1},Q_{L-1};T_{L-1},q_L):q_L \in \overline B_{T_L-T_{L-1}}^{(|S_L|)}(Q_L)\Bigr\}
\ee
to obtain $\Phi(T_1,Q_1,\dots,T_L,Q_L)$; see Figure~\ref{fig:induction}. 

\begin{figure}[htb]
\centering
\includegraphics[width=400pt,keepaspectratio]{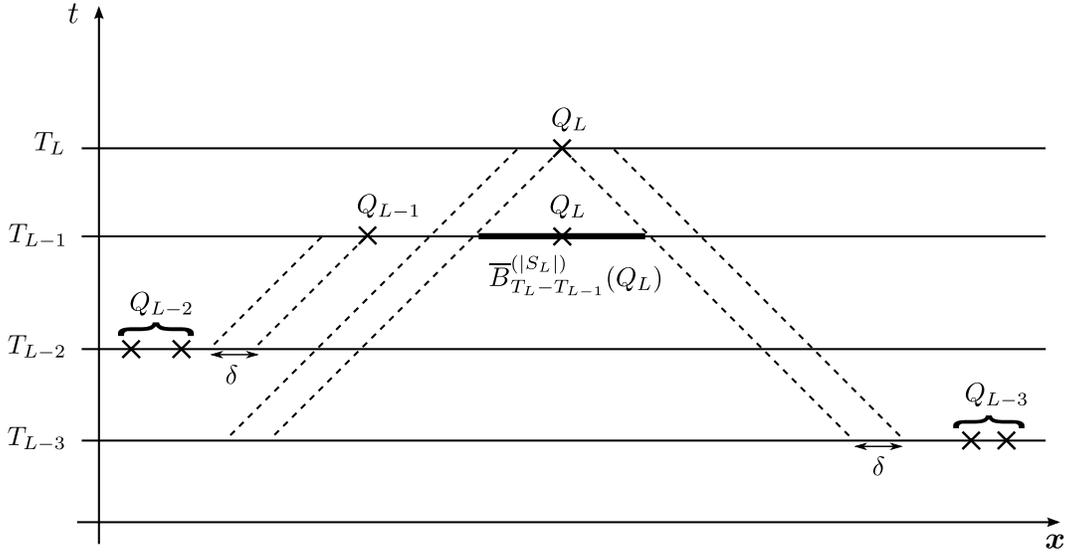}
\centering
\caption{\small{A configuration $(T_1,Q_1,\ldots,T_L,Q_L)\in\sS_{\delta}$, and the set \eqref{DoDsetL} on which $\Phi$ is given before solving \eqref{ind_eq_2} to obtain $\Phi(T_1,Q_1,\ldots,T_L,Q_L)$. In the case drawn, the $L$-th family contains only a single particle (so $Q_L\in\RRR^3$).}}
\label{fig:induction}
\end{figure}

We need to verify that this $\Phi$ is well defined on any overlap $\sS_{\delta,P}\cap \sS_{\delta,P'}$, i.e., that for any $Q^4$ contained not only in $\sS_{\delta,P}$ but also in $\sS_{\delta,P'}$ for a partition $P'$ with $L'<L$ families, the value $\Phi(Q^4)$ just defined agrees with the value defined at a previous round of the induction. Such a $P'$ exists if and only if $T_\alpha=T_{\alpha+1}$ for some $1\leq \alpha<L$. If $T_{L-1}=T_L$ then the solution $\Phi(Q^4)$ of \eqref{ind_eq_2} coincides with the initial condition $\Phi(Q^4)$ given on $\sS_{\delta,\tilde{P}}$, which agrees with the value on $\sS_{\delta,P'}$ by the induction assumption. If $T_\alpha=T_{\alpha+1}$ for some $\alpha<L-1$ then the set \eqref{DoDsetL} is contained not only in $\sS_{\delta,\tilde{P}}$ but already in $\sS_{\delta,\tilde{P}'}$ with $\tilde{P}'$ the partition obtained from $\tilde{P}$ by merging $S_\alpha$ and $S_{\alpha+1}$. By the uniqueness statements of Lemmas~\ref{lemma_chernoff} and  \ref{finite_propagation_speed_N_part_dirac}, the value of $\Phi(Q^4)$ obtained by solving \eqref{ind_eq_2} agrees with the value $\Phi(Q^4)$ already defined using the partition $\{S_1,\ldots,S_{\alpha-1},S_\alpha \cup S_{\alpha+1},S_{\alpha+2},\ldots,S_L\}$. 

\medskip

We also need to verify that $\Phi$ is smooth on $\sS_{\delta,P}$. Apply Lemma~\ref{lemma_chernoff} to an open neighborhood of the set \eqref{DoDsetL} in $\sS_{\delta,\tilde{P}}$, regarding also $T_1,Q_1,\ldots,T_{L-1},Q_{L-1}$ as variables (while the coefficients $A_i(q)$ in \eqref{def_dirac_type_operator} accompanying derivatives relative to these variables vanish); it follows that the solution of \eqref{ind_eq_2} is smooth in a neighborhood of $Q^4$ in $\sS_{\delta,P}$, provided that $T_L>T_{L-1}$. 
The case $T_L=T_{L-1}$ needs separate treatment because varying $T_L$ leads to a change in the numbering of the families $S_1,\ldots,S_L$. So let us fix a numbering and drop the condition $T_1\leq \ldots \leq T_L$. Let $U$ be an open neighborhood in $\sS_{\delta,P}$ of $Q^4$; we focus on a $Q^4\in\sS_{\delta,P}$ with $T_L=T_{L-1}\geq \max\{T_1,\ldots,T_{L-2}\}$. It is clear that, since the families $S_1,\ldots,S_L$ do not interact, the multi-time equations \eqref{PhiHalpha} for $\alpha=1,\ldots,L$ possess a smooth joint solution in $U$, provided $U$ is sufficiently small, from initial data on $\sS_{\delta,\tilde{P}}$ as in \eqref{tildePdef}, using that $\Phi$ on $\sS_{\delta,\tilde{P}}$ satisfies the multi-time equations by induction assumption. This solution agrees with $\Phi$ for $T_L\geq T_{L-1}$ by \eqref{ind_eq_2}, but also for $T_L\leq T_{L-1}$ because in that case we used to renumber $L\leftrightarrow L-1$ before writing down \eqref{ind_eq_2}. Therefore, $\Phi$ must be smooth in a neighborhood of any $Q^4$. 

\medskip

To see that $\phi$ agrees with $\Phi$ on $\sS_{\delta,P}$, note that $\phi$ and $\Phi$ have the same initial conditions on \eqref{DoDsetL} and are solutions to the same equation \eqref{ind_eq_2}; then apply Lemma~\ref{finite_propagation_speed_N_part_dirac}.

\medskip

Now we want to show that $\Phi$ satisfies the multi-time equations \eqref{PhiHalpha} on $\sS_{\delta,P}$. This can be done by comparing $\Phi$ to the unique solution of the multi-time equations from initial data on a sufficiently small open neighborhood $U$ of \eqref{DoDsetL} in $\sS_{\delta,\tilde{P}}$; that solution exists because the appropriate cone over $U$ (i.e., the set with domain of dependence within $U$) lies entirely in $\sS_{\delta,P}$, so that the families $S_1,\ldots,S_L$ do not interact. An alternative route goes as follows. 

By the construction based on \eqref{ind_eq_2}, $\Phi$ satisfies \eqref{PhiHalpha} for $\alpha=L$ (also for $T_L=T_{L-1}$, as shown in the disucssion of smoothness). Now consider $\alpha<L$. Since, on $\sS_{\delta,P}$,
\begin{equation}
\biggl[ \frac{i\partial}{\partial t_\alpha}-H_{S_\alpha},
\frac{i\partial}{\partial t_L} -H_{S_L} \biggr] =0
\end{equation}
because (locally) different families do not interact, we have that
\begin{equation}
\Bigl(\frac{i\partial}{\partial t_L} -H_{S_L}\Bigr) \Bigl(
\frac{i\partial}{\partial t_\alpha} -H_{S_\alpha}\Bigr) \Phi
= \Bigl(\frac{i\partial}{\partial t_\alpha} -H_{S_\alpha}\Bigr) \Bigl(
\frac{i\partial}{\partial t_L} -H_{S_L}\Bigr) \Phi \,,
\end{equation}
and since
\begin{equation}
 \Bigl( \frac{i\partial}{\partial t_L} -H_{S_L}\Bigr) \Phi =0,
\end{equation}
we have that
\be\label{HLHalphaPhi}
\Bigl(\frac{i\partial}{\partial t_L} -H_{S_L}\Bigr) \Bigl(
\frac{i\partial}{\partial t_\alpha} -H_{S_\alpha}\Bigr) \Phi=0.
\ee
We will show that the function
\begin{equation}
\Phi'_{\alpha}:=\Bigl( \frac{i\partial}{\partial t_\alpha} -H_{S_\alpha}\Bigr)
\Phi
\end{equation}
vanishes identically on $\sS_{\delta,P}$. To this end, we note that, by \eqref{HLHalphaPhi},
$\Phi'_{\alpha}$ satisfies \eqref{ind_eq_2} with initial
datum
\begin{multline}
\Phi'_{\alpha}(T_1,Q_1; \dots; T_{L-1},Q_{L-1};T_{L-1},q_L)=\\
\Bigl(\frac{i\partial}{\partial
t_\alpha}-H_{S_\alpha}\Bigr)\Phi(T_1,Q_1; \dots; T_{L-1},Q_{L-1};T_{L-1},q_L).
\end{multline}
This initial datum lies in $\sS_{\delta,\tilde{P}}$ with $\tilde{P}$ as in \eqref{tildePdef}, and by the induction assumption it vanishes identically. By the linearity of \eqref{PhiHalpha} and Lemma~\ref{finite_propagation_speed_N_part_dirac}, also $\Phi'_\alpha$ vanishes identically. 

This completes the induction step and thus the proof.
\end{proof}

\bigskip

\noindent{\it Acknowledgments.} We thank Detlef D\"urr, Felix Finster, Sheldon Goldstein, Michael Kiessling, and Matthias Lienert for helpful discussions. S.P.\ acknowledges support from Cusanuswerk, from the German--American Fulbright Commission, and from the European Cooperation in Science and Technology (COST action MP1006). R.T.\ acknowledges support from the John Templeton Foundation (grant no.\ 37433) and from the Trustees Research Fellowship Program at Rutgers. 

\end{document}